\documentclass[smallcondensed]{svjour3}
\usepackage[margin=1.5in]{geometry}
\usepackage{graphicx}
\usepackage{epstopdf}
\usepackage{amsmath,amssymb,amsfonts}
\usepackage{mathtools}
\usepackage{multirow}
\usepackage{multicol}
\usepackage{verbatim}
\usepackage{url}
\usepackage{comment}
\usepackage{mathtools}
\usepackage{epsf,pgf,graphicx}
\usepackage{color}
\usepackage{tikz,pgf}
\usepackage{diagbox}
\usepackage{makecell}
\usepackage{braket}
\usepackage{mathtools}
\usepackage{xcolor}
\usepackage{subfig}
\usepackage{hyperref}
\usepackage{ulem}
\usepackage{arydshln}
\usepackage[export]{adjustbox}

\usepackage{enumitem}
\newlist{abbrv}{itemize}{1}
\setlist[abbrv,1]{label=,labelwidth=1in,align=parleft,itemsep=0.1\baselineskip,leftmargin=!}

\usepackage{hyperref}
\hypersetup{colorlinks,
citecolor=blue,
urlcolor=blue,
linkcolor=blue
}
\usepackage{tabstackengine}
\stackMath
\usepackage{array}
\newcolumntype{M}[1]{>{\hbox to #1\bgroup\hss$}l<{$\egroup}}

\makeatletter
\newcommand\@brcolwidth{0.67em}

\def\@brarray[#1]{\array{r*\c@MaxMatrixCols {M{#1}}}}
\makeatother

\newtheorem{cons}{Construction}

\title{Almost Perfect Mutually Unbiased Bases that are Sparse}

\makeatletter
\renewcommand*{\@fnsymbol}[1]{\ensuremath{\ifcase#1\or *\or \dagger\or \ddagger\or
   \mathsection\or \mathparagraph\or \|\or **\or \dagger\dagger
   \or \ddagger\ddagger \else\@ctrerr\fi}}
\makeatother

\author{Ajeet Kumar \and Subhamoy Maitra \and Somjit Roy}
\institute{
A. Kumar \at
Applied Statistics Unit, Indian Statistical Institute, Kolkata, India,
\email{ajeetk52@gmail.com}
\and
S. Maitra \at
Applied Statistics Unit, Indian Statistical Institute, Kolkata, India,
\email{subho@isical.ac.in}
\and
S. Roy \at
Department of Statistics, University of Calcutta, Kolkata, India,
\email{somjit.roy2001@gmail.com}
}

\authorrunning{\textsc{Ajeet Kumar}, \textsc{Subhamoy Maitra} \textsc{and Somjit Roy}}

\begin{document}
\maketitle
\begin{abstract}
In dimension $d$, Mutually Unbiased Bases (MUBs) are a collection of orthonormal bases over $\mathbb{C}^d$ such that for any two vectors $v_1, v_2$ belonging to different bases, the dot or scalar product $|\braket{v_1|v_2}| = \frac{1}{\sqrt{d}}$. The upper bound on the number of such bases is $d+1$.  Construction methods to achieve this bound are known for  cases when $d$ is some power of prime. The situation is more restrictive in other cases and also when we consider the results over real rather than complex. Thus, certain relaxations of this model are considered in literature and consequently Approximate MUBs (AMUB) are studied. This enables one to construct potentially large number of such objects for $\mathbb{C}^d$ as well as in $\mathbb{R}^d$. In this regard, we propose the concept of Almost Perfect MUBs (APMUB), where we restrict the absolute value of inner product $|\braket{v_1|v_2}|$ to be two-valued, one being 0 and the other $ \leq \frac{1+\mathcal{O}(d^{-\lambda})}{\sqrt{d}}$, such that $\lambda > 0$ and the numerator $1 + \mathcal{O}(d^{-\lambda}) \leq 2$. Each such vector constructed, has an important feature that large number of its components are zero and the non-zero components are of equal magnitude. Our techniques are based on combinatorial structures related to Resolvable Block Designs (RBDs). We show that for several composite dimensions $d$, one can construct $\mathcal{O}(\sqrt{d})$ many APMUBs, in which cases the number of MUBs are significantly small. To be specific, this result works for $d$ of the form $(q-e)(q+f), \ q, e, f \in \mathbb{N}$, with the conditions $0 \leq f \leq e$ for constant $e, f$ and $q$ some power of prime. We also show that such APMUBs provide sets of Bi-angular vectors which are of the order of $\mathcal{O}(d^{\frac{3}{2}})$ in numbers, having high angular distances among them. Finally, as the MUBs are equivalent to a set of Hadamard matrices, we show that the APMUBs are so with the set of Weighing matrices. 
\end{abstract}
\keywords{Almost Perfect Mutually Unbiased Bases, \ Combinatorial Design, \ Hadamard Matrices, \ Quantum Information Theory, \ Resolvable Block Design, \ Weighing Matrix.}
\subclass{81P68}
\section{Introduction}
\label{Introduction}
Mutually Unbiased Bases (MUBs) received serious attention in Quantum Information Processing, as they are useful in different aspects of Quantum  Cryptology and Communications like Quantum Key Distribution (QKD), Teleportation, Entanglement Swapping, Dense Coding, Quantum Tomography, etc. (see~\cite{Durt2010} and the references therein). For any finite dimensional Hilbert space $\mathbb{C}^d$, the number of MUBs is bounded by $d+1$. However, in spite of intense research for several decades, these could be constructed only when $d$ is some prime power~\cite{Ivonovic1981,Wootters1989,Bandyopadhyay2002,Klappenecker2003,Sulc2007}. For a specific dimension $d$, constructing a larger number of MUBs (upper bounded by $d+1$) is one of the most challenging problems in Quantum Information Theory. 

Various mathematical tools have been used to construct MUBs, among which noteworthy being the use of finite fields~\cite{Wootters1989,Klappenecker2003} and maximal set of commuting bases~\cite{Bandyopadhyay2002}. For dimensions which are not power of primes, constructing large number of MUBs still remains elusive. This is the reason, various kinds of Approximate MUBs have been constructed using character sums over Galois Rings or Galois Fields~\cite{AMUB-Shparlinski20062006,ASIC-Klappenecker2005,AMUB-MixedCharacterSum,AMUB-CharacterSum,AMUB-FrobeniusRing,AMUB-GaloisRing,ARMUB-fromComplexAMUB}, combinatorial design ~\cite{AK20,ak22} and computational search~\cite{AK22}.

When MUBs are constructed over $\mathbb{R}^d$, we get Real MUBs. They have interesting connections with Quadratic Forms~\cite{QuadraticForm}, Association Schemes~\cite{ReMUB-AssociationSchemes,AssociationScheme-Coding}, Equi-angular Lines,  Equi-angular Tight Frames over $\mathbb{R}^d$~\cite{ReMUB-FusionFrames}, Representation of Groups~\cite{ReMUB-GroupRepresentation}, Mutually Unbiased Real Hadamard Matrices, Bi-angular vectors over $\mathbb{R}^d$~\cite{Holzmann2010,Kharaghani2018,Best2015} and Codes ~\cite{Calderbank1997}. As we have elemented out earlier, large number of Real MUBs are non-existent for most of the dimensions~\cite{boykin2005real}. In fact only for $d = 4^s , s > 1$, we have $d/2 + 1$ many MUBs, whereas for most of the dimensions $d$, which are not perfect square, we have at best only $2$ Real MUBs~\cite{boykin2005real}. In view of this, various attempts have been made to construct Approximate Real MUBs (ARMUBs) which are available in literature~\cite{AK20,ak22,ARMUB-fromComplexAMUB}. 

Various efforts have been made to explore connections between MUBs and geometrical objects such as {polytopes} and {projective planes}~\cite{MUB-Polytopes,MUB-CompPolytopes,MUBs-ProjectivePlanes,MUBs-HjelmslevGeometry,GaloisUnitaries-MUB}. Since the known methods for the construction of MUBs provides complete sets only when $d$ is some power of prime, there are conjectures related to the existence of complete sets of MUBs and finite projective plane, which are also currently known to exist only for prime power orders. If $d = p_1^{n_1} p_2^{n_2} \ldots p_s^{n_s}$, then the lower bound on the number of MUBs is $p_r^{n_r}  +1$ where $p_r^{n_r} = \min\{{p_1^{n_1}, p_2^{n_2}, \ldots,p_s^{n_s}}\}$. Thus, constructing a large number of MUBs for any composite dimension has proven to be elusive even over $\mathbb{C}^d$. In fact, the number of such bases is very small when we consider the problem over the real vector space $\mathbb{R}^d$ (see~\cite{boykin2005real}). For discussions on the basics and open problems related to MUBs and the approximate version, one may refer to our earlier works~\cite{AK20,ak22} and the references therein. 

Towards constructing the approximate MUBs over $\mathbb{C}^d$, certain techniques were proposed in~\cite{AMUB-Shparlinski20062006,ASIC-Klappenecker2005,AMUB-CharacterSum}. However, such techniques cannot be applied in construction of Real MUBs. In this direction, we could propose a construction in~\cite{AK20} for Approximate Real MUBs using Real Hadamard matrices. It has been shown in~\cite{AK20} 
that $\frac{\sqrt{d}}{4}+1$ ARMUBs with maximum value of the inner product as $\frac{4}{\sqrt{d}}$ could be achieved for $d = (4q)^2$, where $q$ is a prime. In~\cite{ak22}, we have shown that such results can be generalized as well as improved to a great extent using Resolvable Block Designs (RBDs). The earlier result of~\cite{AK20} could be generalized in~\cite{ak22} for $d = sq^2$, where $q$ is a prime power. Further, the parameters could be improved too in~\cite{ak22}. It has been shown in~\cite{ak22} that for $d = q(q+1)$, where $q$ is a prime power and $q \equiv 3 \bmod 4$, it is possible to construct $\lceil{\sqrt{d}}\rceil = q+1$ many ARMUBs with maximum value of the inner product upper bounded by $\frac{2}{\sqrt{d}}$, between vectors belonging to different bases. Therefore, the improvement in the result of~\cite{ak22} is two-fold. First, the number of MUBs is greater and second, the maximum of the inner product values is lower compared to~\cite{AK20}. As these were achieved by several kinds of combinatorial designs, we explore this idea further. One important feature of all our construction techniques~\cite{AK20,ak22} was that, all the components of the basis vectors so constructed is either zero or of constant magnitude. Hence taking the normalizing factor outside the vectors would render the numerical value of the components as zero or of unit magnitude, which is like Weighing matrices~\cite{weigh1} -- a generalization of Hadamard matrices. Finally, note that the construction of~\cite{ak22} provides the vectors which are sparse and this property is inherited in this paper too. 

Approximate MUBs have been defined in various manners in the state of the art literature. The cue has been taken from the two initial papers~\cite{ASIC-Klappenecker2005,AMUB-Shparlinski20062006}. Although the work of~\cite{ASIC-Klappenecker2005} is related to Approximate SIC POVM, the definition of ``Approximate" has been carried over to MUBs as well. Various mathematical meanings of approximations have been used in relaxing the condition on absolute value of the inner product between two vectors say $\ket{v_1}, \ket{v_2}$. For example, in these two papers~\cite{ASIC-Klappenecker2005,AMUB-Shparlinski20062006}, we get references to $|\braket{v_1|v_2}|$ bounded by  $\frac{1+o(1)}{\sqrt{d}}$, $\frac{2+o(1)}{\sqrt{d}}$, $\mathcal{O}(\frac{log(d)}{\sqrt{d}})$,  $\mathcal{O}(\frac{1}{\sqrt[4]{d}})$, and $\mathcal{O}(\frac{1}{\sqrt{d}})$ etc. Subsequent researchers investigating Approximate MUBs have also adopted similar mathematical definition~\cite{AMUB-CharacterSum,AMUB-MixedCharacterSum,AMUB-FrobeniusRing,AMUB-GaloisRing,ARMUB-fromComplexAMUB}.

One may note that each MUB in the space $C^{d}$ consists of $d$ orthogonal unit vectors which, collectively, can be thought of as a unitary $d \times d$ matrix. Two (or more) MUBs thus correspond to two (or more) unitary matrices, one of which can always be mapped to the identity $I$ of the space $C^{d}$, using a unitary transformation. For example, suppose we have $r$ many MUBs $\{M_{1}, M_{2}, M_{3}, \dots, M_{r}\}$ in $C^{d}$ where $r \leq d+1$ and also we can thought them as a $r$ numbers of $d \times d$ unitary matrices.
If we multiply $M_{1}^{-1}$ to each of the matrices at right, then one can obtain $\{I, M_2 M_{1}^{-1}, M_3 M_{1}^{-1}, \dots, M_r M_{1}^{-1}\}$ as the transformed set of MUBs. As the inverse of any unitary matrix is equal to its conjugate transpose, to obtain $M_j M_{i}^{-1}$, for $i \neq j$,
we are considering inner products of each row of the two matrices in the set of MUBs. Thus, the modulus of each element of the product matrix will be $\frac{1}{\sqrt{d}}$. Taking $\frac{1}{\sqrt{d}}$ common, the modulus of each of the elements will be 1, i.e., we will have complex Hadamard matrices. The result will be similar if we multiply the inverse from the left too. In this regard, we now consider the Weighing matrices.

\begin{definition}
A square matrix of order $d$ and weight $w$ is called a (complex) weighing matrix, denoted by $W(w,d)$, if its elements belong to the set $\left\{0, \frac{\exp(i \theta)}{\sqrt{w}} \right\}$ with $\theta \in \mathbb{R}$, and it satisfies $W(w,d)^\dag W(w,d) = I$. If the elements are confined to the set $\left\{0, \pm \frac{1}{\sqrt{w}}\right\}$, it becomes a real weighing matrix.
\end{definition}
 
The use of complex weighing matrices in quantum error correcting codes has been explored in \cite{egan2023survey}. Furthermore, the connection between real weighing matrices and classical codes are also investigated, as evident from the studies such as \cite{crnkovic2021lcd,jungnickel1999perfect,jungnickel2002perfect,arasu2001self}, which also delve into applications involving spherical codes \cite{nozaki2015weighing}. For more analysis, one can refer to \cite{koukouvinos1997weighing} and the references therein. As we will proceed further, the definition of APMUBs will be presented and we will show that the weighing matrices will related to APMUBs as the Hadamard matrices relate to MUBs.

Let us now summarize the contribution of this paper and outline its presentation. 
\subsection{Organization and Contribution}
\label{Contribution}
We begin with Section~\ref{Background} to present a background of related combinatorial objects. Then towards the constructions, in Section~\ref{Lemmas}, we show bounds on the values of certain parameters, expressed in terms of the block size $k$ and number of elements in the RBD, i.e., $d$. In this regard, we define a combinatorial quantity $\mathcal{A}(d,k,\mu)$, relevant to our analysis, which can be of its own independent interest.  Thereafter in Lemma~\ref{MOLS-RBD} we describe an interesting class of RBDs which can be constructed from MOLS(s) yielding $\mu = 1$ and $r = N(s) +2$. Here $\mu$ is the maximum number of common elements between any pair of blocks from different parallel classes, as we will explain in the following background section. A constructive proof to obtain the same from MOLS(s) has also been given and in Lemma~\ref{RBD-MOLS}. We further show that the converse of Lemma~\ref{MOLS-RBD} is also true. The results are further explained with illustrative examples. 

Then we consider the Almost Perfect MUBs (APMUBs) and some generic ideas of construction in Section~\ref{apmub}.
The basic motivation and its relationship with Bi-angular vectors are presented in Section~\ref{Motivation}. 
In Section~\ref{TheoreticalAnalysis}, we analyze certain properties of the AMUBs which can be constructed using RBDs having blocks of constant size. In this direction, we consider RBD$(X,A)$ with $|X| = d = k \cdot s = (q-e)(q+f)$ and $A$ with resolution $r$, where each block is of size $(q-e)$. We study the asymptotic behaviours of the parameters of AMUBs thus generated. It is also shown that our construction can provide APMUBs only when $\mu =1$, therefore putting strong constraints over the nature of the RBDs required for this kind of constructions.

Section~\ref{ConsAPMUB} contains our algorithms towards constructing RBDs that can be consequently used for the obtaining APMUB's with parameters that could be achieved for the first time. We first show that whenever the dimension $d$ is a composite number and can be expressed as $k \cdot s, k\leq s$, such that $\beta = \sqrt{\frac{s}{k}}\leq 2$, one can construct $N(s)+1$ many APMUBs, where $N(s)$ is the number of MOLS$(s)$. We refer to this as the MOLS Lower Bound Construction for APMUBs. Since a composite number $d$ can be factored in multiple ways ensuring $\beta \leq 2$, there can be more than one  MOLS Lower Bound Constructions for a dimension $d$. It is to be noted that if $s = q$, some power of prime, then $N(q) = q-1$. Hence in such situations we get $q$ many APMUBs. The best known asymptotic bound for $N(s)$ is given by $N(s) \rightarrow \mathcal{O}(s^{\frac{1}{14.8}})$, which generally results into a  small number of APMUBs. We improve this significantly. 

In this direction we show that when $d = (q-e)(q+f), e, f \in \mathbb{N}$ and $e \geq f$, where $q$ is some power of prime, we can obtain $\mathcal{O}(q)$ many APMUBs. That is, when $e, f$ are constant, then the number of such APMUBs is $\mathcal{O}(\sqrt{d})$.  Illustrative examples to describe the above construction have also been provided. This is explained in Section~\ref{case2}.
  
We conclude with the following two comments in Section~\ref{conclusion} towards future direction. First, the dimensions $d$ for which we obtain $O(\sqrt{d})$ many MUBs have the density of the order of the primes in the space of natural numbers. Secondly, the constructions of APMUBs directly translates to the construction of Bi-angular vectors. The other direction may also be explored. It needs further disciplined effort to explore this part too.  

\section{Background}
\label{Background} 
Let us start with the basic definition related to Mutually Unbiased Bases.
\begin{definition}
Consider two orthonormal bases, 
$$M_l = \left\{\ket{\psi_1^l}, \ket{\psi_2^l}, \ldots, \ket{\psi_{d}^l}\right\} 
\mbox{ and } M_m = \left\{\ket{\psi_1^m}, \ket{\psi_2^m}, \ldots, \ket{\psi_{d}^m}\right\},$$ 
in $d$-dimensional complex vector space, i.e., $\mathbb{C}^{d}$. These two bases will be called {Mutually Unbiased} if we have
\begin{equation}
    \left|\braket{\psi_i^l | \psi_j^m}\right| = \displaystyle\frac{1}{\sqrt{d}}, \ \forall i, j \in \left \{1, 2, \ldots, d\right\} .   
\end{equation}
\end{definition}
The set $\mathbb{M} = \left\{M_1, M_2, \ldots, M_r\right\}$ consisting of such orthonormal bases will form MUBs of size $r$, 
if every pair in the set is mutually unbiased.

When the conditions among two different bases are relaxed such that $\left|\braket{\psi_i^l | \psi_j^m}\right|$ can take values other than $\frac{1}{\sqrt{d}}$, then we consider the approximate version of 
this problem.   This is due to the fact that, for most of the dimensions, which are not power of some primes, obtaining a large number of MUBs reaching the upper bound is elusive. In this context we denote $\Delta$ for the set containing different   values of $\left|\braket{\psi_i^l | \psi_j^m}\right|$ for $l \neq m$. In this initiative, extending the ideas of~\cite{ak22}, we exploit the well known combinatorial object, the Resolvable Block Design (RBD), towards construction of Approximate MUBs with improved parameters. One may refer to the book ~\cite{stinson2007combinatorial} for more details on RBDs, and we present certain definitions in this regard. 
\begin{definition}
A combinatorial block design is a pair $(X, A)$, where $X$ is a set of elements, called elements, and $A$ is a collection of non-empty subsets of $X$, called blocks. A combinatorial design is called {simple}, if there is no repeated block in $A$.
\end{definition}
Generally all the combinatorial designs are assumed to be simple, i.e., they do not have any repeated blocks. 
\begin{definition}
\label{def3x}
A combinatorial design $(X,A)$ is a $t-(d,k,\lambda)$ design if each block in $A$ is of size $k$, and that any set of $t$ elements from $X$, appears as subset of exactly $\lambda $ blocks in $A$. Note that here $t, d, k$ and $\lambda$ are positive integers with $1< k< d$.
\end{definition}

\subsection{Resolvable Block Design (RBD)}
Resolvable Block Design (RBD) is a special kind of Combinatorial design, where the set $A$ can be partitioned into parallel classes which are called resolutions of $A$. Initially RBD was defined by R. C. Bose in the context of Balanced Incomplete Block Design~\cite{Bose1939,Bose1942,Bose1947}. Later various generalizations could be achieved as explained in varied literatures~\cite{Shrikhande1965,Kageyama1976,Patterson1976,John1999}. For the purpose of our paper we consider the following simple definition of RBD as presented in our previous paper ~\cite{ak22}.

\begin{definition}
Combinatorial design $(X, A)$, is called a Resolvable Block Design (RBD), if $A$ can be partitioned into $r \geq 1$ parallel classes, called resolutions. Where a parallel class in design $(X, A)$ is a subset of the disjoint blocks in $A$ whose union is $X$. 
\end{definition} 

There is a special kind of RBD called Affine Resolvable BIBD (ARBIBD)~\cite{Bose1942,Bose1947,Shrikhande1976} (see also~\cite[Chapter 5]{stinson2007combinatorial}). It is well known that whenever $q$ is some power of a prime, one can construct $(q^2,q,1)$ ARBIBD. An Affine plane of order $q$ is an example of this. Here $|X| = q^2$, and $A$ consists of $q(q+1)$ blocks, which can be resolved into $q+1$ many parallel classes. Each parallel class consists of $q$ many blocks of constant size $q$. Most importantly, any pair of blocks from different parallel classes has exactly one element in common. Affine Planes are known only when $q$ is power of a prime. For detail, one may refer to~\cite[Sections 2.3, 5.2, 5.3, 6.4]{stinson2007combinatorial}.
 
Let us define the notation that we will use in this connection. For RBD$(X,A)$, with $|X| = d$, we will indicate the elements (also called elements) of $X$ by simple numbering, i.e., $X = \{1, 2, 3, \ldots, d\}$. Here $r$ will denote the number of parallel classes in RBD and parallel class will be represented by $P_1, P_2, \ldots, P_r$. The blocks in the $l^{th}$ parallel class will be represented by $\{b_1^l, b_2^l, \ldots, b_s^l \}$, indicating that the $l^{th}$ parallel class has $s$ many blocks.  Since in our entire analysis we will be using RBDs with constant block size, let us denote the block size by $k$. Further, we denote the number of blocks in a parallel class of RBD by $s$. Since in our analysis we are making of use of RBDs with constant block size, hence each parallel class will always have $s$ many blocks and $|X| = d = k \cdot s$.  The notation $b_{ij}^l $ would represent the $j^{th}$ element of the $i^{th}$ block in the $l^{th}$ parallel class. Further, the notation $b_{i}^l $ would represent  $i^{th}$ block of $l^{th}$ parallel class. Note that $b_{ij}^l \in X$. In every block, we will arrange the elements in increasing order, and we will follow this convention throughout the paper, unless mentioned specifically. Thus $b_{ij}^l \leq b_{i,j+1}^l, \ \forall j$. This will be important to revisit when we convert the parallel classes into orthonormal bases. Another important parameter for our construction is the value of  the maximum number of common elements between any pair of blocks from different parallel classes. We denote this positive integer by $\mu$. Note that $\mu \geq 1$ for any RBD, with $r \geq 2$. One may further refer to Lemma \ref{RBD1} of Section \ref{Lemmas} in this regard. 

\subsection{Mutually Orthogonal Latin Square (MOLS)}
A Latin Square of order $s$ is an $s \times s$ array, and a cell of the array consists of a single element from a set $Y$, such that $|Y| = s$. Every row of the Latin Square is a permutation of  the elements of set $Y$ and every column of the Latin square is also  permutation of the elements from the set $Y$. For more details one may refer to~\cite[Definition 6.1]{stinson2007combinatorial} as well as~\cite[Example 1.1]{HandBook}. A pair of Latin Squares $(L_1,L_2)$ of same order and having entries from same set $Y$ (or, a different set having same number of elements) is called Mutually Orthogonal, if in the ordered pair  $\{(Y,Y)\} = \{((L_1)_{ij},(L_2)_{ij})\}$, every  pair $x, y \in Y$ appears exactly once. That is, if two of the Latin Squares are superimposed, and the resulting entries in each cell is written as ordered pairs, then every  $x,y \in Y$ appears exactly once in the cell.  Further, if there is a set of $w$ many Latin Squares, say $\{L_1,L_2,\ldots,L_w\}$, each of order $s$,  such that, every pair of Latin Squares is orthogonal, then the set is called Mutually Orthogonal Latin Square of order $s$, which we denote as $w$-MOLS$(s)$. 

Let $N(s)$ denote the maximal value of $w$ such that, there are $w$ many MOLS of order $s$~\cite{HandBook,Colbourn2001},~\cite[Chapter 6]{stinson2007combinatorial}. While using the numerical values of $N(s)$, in subsequent examples in this paper, we will use the currently known values of $N(s)$ from~\cite[Table 3.87,  page 176]{HandBook}. Note that these are not always the actual values of $N(s)$ (except when $s$ is some power prime or of small order) as the exact value of $N(s)$ is still an open question in most of the cases.  It is known that, $N(s) \leq s - 1$ $\forall \ s$. When this bound is attained, we say that there is a complete set of Mutually Orthogonal Latin Squares of order $s$. The construction for complete sets of MOLS$(s)$ is known  when $s$ is some power of prime~\cite[Section 6.4]{stinson2007combinatorial}. When $s$ is not a power of prime, $N(s)$ is much smaller than $s-1$. A table with the largest known values for $w$ is presented in~\cite{HandBook} for $s < 10000$.

It is known that there exists a constant $n_0$, such that for all $s \geq n_0$, we have, $N(s) \geq \frac{1}{3} s^{\frac{1}{91}}$~\cite{Chowla1960}, which was later improved by Wilson~\cite{Wilson1974} to $N(s) \geq s^{\frac{1}{17}}$.  Further, it was shown in~\cite[Section 4]{wocjan2004new} that the exponent can be lower bounded by $\frac{1}{14.8}$. One may note that $N(s) \rightarrow \infty$ as $s \rightarrow \infty$ in general, but for the finite cases only when $s$ is some power of prime then $N(s)= s-1$, else it is considerably small. In this regard, one may also note that the Affine Planes of order $q$ are equivalent to $(q-1)$ MOLS$(q)$~\cite[Theorem 6.32]{stinson2007combinatorial}. For a brief survey on construction of MOLS, one may refer to~\cite{Colbourn2001}. Note that, in connection of RBD, we used $s$ to indicate the number of blocks in a parallel class, whereas in the context of Latin Square, we use $s$ to indicate the order of Latin Square.

\section{Some Important Technical Results}
\label{Lemmas}
Let us now consider a counting of combinatorial objects that is relevant to us. 

\begin{definition}
Let ${\cal{T}}(d,k,\mu)$, $0 \leq \mu < k < d$ be the maximum number of subsets each of size $k$, that can be constructed from $d$ distinct objects, such that, between any two different subsets there is a maximum of $\mu$ objects in common. 
\end{definition}
First we should relate with the error correcting codes, as this can be seen as  the maximum number of codewords of the binary constant weight codes of length $d$ and weight $k$ with minimum distance $2(k - \mu)$. For more details in this regard one may refer to~\cite[Theorem 2.3.6]{huff}, but we will only restrict here to some technical results only.
One may immediately note that ${\cal{T}}(d, k,0) = \left\lfloor\frac{d}{k} \right\rfloor$, and ${\cal{T}}(d, k,k-1) = \binom{d}{k}$.
For arbitrary $d, k, \mu \in \mathbb{N}$, the following result provides an estimate of ${\cal{T}}(d,k,\mu)$.

\begin{lemma}
\label{A-d-k-mu}
${\cal{T}}(d,k,\mu) \leq \left\lfloor \frac{\binom{d}{\mu+1}}{\binom{k}{\mu+1}} \right\rfloor = \left\lfloor \frac{d! (k-\mu-1)!}{k! (d-\mu-1)!} \right\rfloor$. The upper bound of $\left\lfloor \frac{d! (k-\mu-1)!}{k! (d-\mu-1)!} \right\rfloor$ is achieved whenever $(\mu+1)-(d,k,1)$ design exists and in such cases ${\cal{T}}(d,k,\mu)$ is the same as the number of blocks in $(\mu+1)-(d,k,1)$ design as in Definition~\ref{def3x}.
\end{lemma}
\begin{proof}
Given a set of $d$ distinct elements, we like to construct the maximum number of subsets each of size $k$, such that any two subsets has $\mu$ elements in common. Let us label these blocks as $\{b_1, b_2, \ldots b_r\}$, with $r = {\cal{T}}(d,k,\mu)$. 

Now consider all the $(\mu+1)$-element subsets of $d$ distinct elements. They will be $\binom{d}{\mu+1}$ in numbers. Now consider the blocks $b_i$ and $b_j$. Since there are a maximum of $\mu$ elements in common, any $(\mu+1)$-element subset of $d$ elements cannot exist, which occur in both the blocks $b_i$ and $b_j$. Since each block $b_i$ is of size $k$, the number of $(\mu+1)$-element subsets which can be constructed by the $k$ elements in $b_i$ is $\binom{k}{\mu+1}$. Let us denote this set by $S_i$. Similarly for block $b_j$ the number of $(\mu+1)$-element subsets which can be constructed with its $k$ elements  is $\binom{k}{\mu+1}$. Let us denote this set by $S_j$. We have already seen $S_i \cap S_j = \phi$, else there would be $\mu+1$ element common between $b_i$ and $b_j$. Now since there are $r$ such blocks each of size $k$, hence $|S_1|+|S_2|+\ldots |S_r| = r \binom{k}{\mu+1}$. This must be less than or equal to  $\binom{d}{\mu+1}$, which is the maximum possible $(\mu +1)$-element subsets that can be constructed from $d$ distinct elements. This implies $r \binom{k}{\mu+1} \leq \binom{d}{\mu+1} \Rightarrow {\cal{T}}(d,k,\mu)= r \leq \left\lfloor \frac{\binom{d}{\mu+1}}{\binom{k}{\mu+1}} \right\rfloor = \left\lfloor \frac{d! (k-\mu-1)!}{k! (d-\mu-1)!} \right\rfloor$. 

To see the second part of the lemma, note that $t-(d,k,1)$ design is a design $(X,A)$ where $A$ contains the subsets of $X$ called blocks, such that $|X| = d$ and each block contains exactly $k$ elements. Every $t$-element subset of $X$ is contained in exactly one block. Hence this implies that any two blocks of the design has maximum $t-1$ elements in common. Thus it immediately follows that if $(\mu+1)-(d,k,1)$ design exists, then blocks of the design satisfies the property of ${\cal{T}}(d,k,\mu)$, and since number of blocks in $(\mu+1)-(d,k,1)$ design is $\frac{\binom{d}{\mu+1}}{\binom{k}{\mu+1}}$  ~\cite[Chapter 9, Theorem 9.4 and the following observation]{stinson2007combinatorial},  which is exactly the upper bound of ${\cal{T}}(d,k,\mu)$ as proven above.
\end{proof}

The upper bound of ${\cal{T}}(d,k,\mu)$ is achieved whenever a $(\mu+1)-(d,k,1)$ design exists, which is known for many values of $0\leq \mu < k < d$. This implies that the bound for ${\cal{T}}(d,k,\mu)$ given by the above result is tight. However, getting an exact value/expression  of ${\cal{T}}(d,k,\mu)$ appears to be an open and challenging problem.
 
To construct the AMUBs, our focus has been on RBDs with constant block size. In this connection we now focus on few results that are relevant to our constructions of APMUBs in Section \ref{ConsAPMUB}. Following is a lemma related to an RBD providing a bound for $\mu$ and $r$ in terms of the block size and the number of blocks in parallel classes where, as defined previously, $\mu$ is the maximum number of common elements between any pair of blocks from different parallel classes and $r$ is the number of parallel classes.  
\begin{lemma}
\label{RBD1}
Consider an RBD$(X,A)$ with $|X| = d = k \cdot s$ where $ k, s \in \mathbb{N}$, consisting of $r>1$ parallel classes, each having blocks of size $k$. Then $\mu \geq \lceil \frac{k}{s} \rceil$, where $\mu$ is the maximum number of common elements between any pair of blocks from different parallel classes and $r \leq {\cal{T}}(d-1,k-1,\mu-1)$. Further, if $\mu = 1$, then $r \leq \lfloor\frac{d-1}{k-1}\rfloor = s + \lfloor\frac{s-1}{k-1}\rfloor$.
\end{lemma}
\begin{proof}
Since $r>1$, consider any pair of parallel classes of RBD say $(P_l, P_m)$. Denote the blocks of $P_l$ be as $b^l_1,b^l_2,\ldots,b^l_s$. Since blocks are of constant size $\Rightarrow |b^l_i|=k$ and the blocks belonging to the same parallel class have no element in common, $\Rightarrow b^l_i \cap b^l_j = \phi \,\, \forall \,i,j= 1, 2, \ldots, s$ and $X = b^l_1 \cup b^l_2 \cup \ldots \cup b^l_s$, $|X| = k\cdot s$. Similar relations will hold for blocks of any other parallel class. 

Consider any block of $P_l$, say $b^l_i$. Since, $\mu = \max_m{|b^l_i \cap b^m_j|}$, we have $b^l_i \cap b^m_j \leq \mu, \ \forall \ j = 1, 2, \ldots, s$.  Since $X =b^m_1 \cup b^m_2 \cup \ldots \cup b^m_s$, hence, $\sum_{j=1}^{s} |b^l_i \cap b^m_j| = |b^l_i| = k$. We also have $\sum_{m=1}^{s} |b^l_i \cap b^m_j| \leq  \sum_{m=1}^{s}\mu = \mu s \Rightarrow k \leq \mu s$. Since $k, s, \mu \in \mathbb{N}$, we get $\mu \geq \lceil \frac{k}{s} \rceil$. This implies minimum value of $\mu = 1$ which is possible only if $k \leq s$ and on the other hand if $k > s$, then minimum value of $\mu = 2$. Thus, for $\mu = 1$, we must have $k \leq s$, i.e., number of blocks must be greater than or equal to the block size of the RBD.
 
To obtain a bound on $r$, fix an element, say $x \in X$. Since the blocks of a parallel class are mutually disjoint and their union is $X$, each parallel class will have exactly one block which will contain $x$. Collect all the blocks that contain the element $x$ and we will obtain a set of $r$ such blocks. Denote this set by $S$. Now, remove $x$ from every block in the set $S$. Hence $S$ will consist of blocks of size $k-1$, the maximum number of common elements between any two blocks in $S$ would be now $(\mu-1)$ and the total number of elements contained in these $r$ blocks will be $\leq (d-1)$. Therefore, we obtain $r \leq {\cal{T}}(d-1,k-1,\mu-1)$. Thus if $\mu = 1$ we have $r \leq {\cal{T}}(d-1,k-1,0) = \lfloor \frac{d-1}{k-1} \rfloor = \lfloor\frac{s \cdot k - 1}{k-1}\rfloor = s+ \lfloor\frac{s-1}{k-1}\rfloor$. 
\end{proof}
We will see that, using our construction method for APMUBs, the necessary condition on RBDs is $\mu=1$. Hence our effort will be to construct RBDs with $\mu = 1$. For this to happen, $k \leq s$ and $r \leq s + \lfloor\frac{s-1}{k-1}\rfloor$. 

\subsection{Results relating to MOLS}
We will now consider a class of Resolvable Block Designs $(X, A)$ such that $|X| = s^2$, which can be constructed from a set of $w$-MOLS$(s)$. Here $A$ consists of blocks having constant size $s$, which can be resolved into $w+2$ many parallel classes, each having $s$ many blocks, such that blocks from two different parallel classes have exactly one element in common. Through the following construction, we explain a simple and direct way to convert a set of $w$ many MOLS$(s)$ into such an RBD$(X, A)$.

\begin{cons}
\label{cons:MOLS-RBD}
To construct an RBD having $w+2$ number of parallel classes from $w$-MOLS$(s)$ using the following steps.
    
\begin{enumerate}
\item Define $M_{ref}$, which is a $s \times s$ array where each cell consists of one of the elements from $X = \{1, 2, \ldots, s^2\}$, as follows:
{\small
\begin{equation*}
 M_{\text{ref}} = \begin{bmatrix}
 1 & 2 & 3 & \ldots & s \\
  s+1 & s+2 & s+3 & \ldots & 2s\\
  2s+1 & 2s+2 & 2s+3 & \ldots & 3s\\
  \vdots & \vdots & \vdots & \ddots & \vdots\\ 
 (s-1)s+1 \,\,&\,\, (s-1)s+2 \,\,&\,\,(s-1)s+3\,\,&\ldots &\, s^2 \\ 
 \end{bmatrix},
  L_k =\begin{bmatrix}
  l_{11}^k & l_{12}^k & \ldots & l_{1s}^k\\\\
    l_{21}^k & l_{22}^k & \ldots & l_{2s}^k \\
    \vdots & \vdots & \ddots & \vdots\\ 
    l_{s1}^k & l_{n2}^k & \ldots &l_{ss}^k
    \end{bmatrix} 
\end{equation*}
}

\item Consider a Latin Square $L_k$, from the set of $w$-$MOLS(s)$. Let  $(L_k)_{ij} = l^k_{ij}$ as indicated above.

\item  Corresponding to the Latin Square $L_k$, construct a parallel class $P_k$ consisting of $s$ disjoint blocks $b^k_t$, each of size $s$ as follows,
    \begin{equation*}
        b_t^k = \{ (M_{ref})_{ij} : l^k_{ij} = t \}, \text{ where } i,j \in \{1, 2, \ldots, s\}.
    \end{equation*}
    
Each row of the the Latin Square is a permutation of $\{1, 2, \ldots, s\}$. Hence, there will be a pair $(i, j)$ in each row for which $l^k_{ij} = t$. Thus, the blocks $P_t^k$ will have a total $s$ elements, one from each row and column of $M_{ref}$. Since $t = \{1, 2, \ldots, s\}$, there will be $s$ blocks. Thus we are essentially collecting all the elements of $M_{ref}$ corresponding to a particular symbol $t$ of $L_k$ in one block $b^k_t$, and together the blocks $b^k_i$ where $i=\{1, 2, \ldots, s\}$ form a parallel class $P_k$.
    
\item Repeat the above step for all Latin Squares in the set of  $w$-$MOLS(s)$, thereby giving $w$ many parallel classes.    
    
\item Construct two more parallel classes, one using the horizontal rows of $M_{ref}$, and other using the vertical rows of $M_{ref}$ as follows:
 \begin{align*}
 P_0 &= \left\{ (1,2,\ldots,s), (s+1,s+2,\ldots, 2s ),\ldots ( (s-1)s+1,(s-1)s+2, \ldots, s^2) \right\}  \\
 P_\infty &=\left\{(1,s+1,\ldots,(s-1)s+1), (2,s+2,\ldots,(s-1)s+2),\ldots,(s,2s,\ldots,s^{2}) \right\}
\end{align*}  
\item The $RBD(X,A)$ with $X = \{1, 2, 3, \ldots, s^2\}$ and $A = \{P_0, P_\infty, P_1, P_2, \ldots, P_w \}$ is the desired outcome.
\end{enumerate}
\end{cons}

\begin{lemma}
\label{MOLS-RBD}
A set of $w$ many MOLS$(s)$ can be used to construct an RBD$(X, A)$ such that $|X| = s^2$, consisting of constant block sizes, each having $s$ elements, that can be resolved into $w+2$ many parallel classes. Here, any two blocks from different parallel classes will have exactly one element in common. 
\end{lemma}
\begin{proof} 
We claim that, any pair of  blocks from different parallel classes $\{P_0 ,P_\infty ,P_1,P_2,\ldots P_w  \}$ constructed above has exactly one element in common. Consider the $t^{th}$ and $s^{th}$ blocks of $k^{th}$ and $m^{th}$ parallel classes respectively. Then
    \begin{equation*}
        P_t^k \cap P_s^m = \{ (M_{ref})_{ij} : l^k_{ij} =t \} \cap  \{ (M_{ref})_{ij} : l^m_{ij} = s\}.
    \end{equation*}
Now since $L^k$ and $L^m$ are the orthogonal Latin Squares, there will be exactly one pair $(i,j)$ such that, $(L^k)_{ij} = t$ and $(L^m)_{ij} =s$. Hence exactly one element will be common between the blocks $P_t^k $ and $P_s^m$.
    
From the construction of $P_0$ and $P_\infty$, it is clear that any block has exactly one element in common. Since any block of $P_t^k$ is picking one element from each row and each column of $M_{ref}$, each block of $P_t^k$ will have exactly one element in common with blocks of $P_0$ and $P_\infty$, which are collection of horizontal rows $(P_0)$ and vertical rows $(P_\infty)$ of $M_{ref}$. 
\end{proof}

We will now sketch the idea of the above method with a simple example to convert a $2$-MOLS$(5)$ into $4$ parallel classes. 
\begin{example}\label{Example-Lemma3}
Let us consider the $2$-MOLS$(5)$ and $M_{ref}$ as follows. 

{\tiny
$ L_1 =\begin{bmatrix}
    5 & 1 & 2 & 3 & 4\\
    1 & 2 & 3 & 4 & 5\\
    2 & 3 & 4 & 5 & 1\\
    3 & 4 & 5 & 1 & 2\\
    4 & 5 & 1 & 2 & 3\\
    \end{bmatrix},
    L_2 =
    \begin{bmatrix}
    5 & 1 & 2 & 3 & 4\\
    2 & 3 & 4 & 5 & 1\\
    4 & 5 & 1 & 2 & 3\\
    1 & 2 & 3 & 4 & 5\\
    3 & 4 & 5 & 1 & 2\\
    \end{bmatrix},
    M_{ref} =
    \begin{bmatrix}
    1 & 2 & 3 & 4 & 5\\
    6 & 7 & 8 & 9 & 10\\
    11 & 12 & 13 & 14 & 15\\
    16 & 17 & 18 & 19 & 20\\
    21 & 22 & 23 & 24 & 25\\
    \end{bmatrix}
$.
}

We use a $5 \times 5$ Reference Matrix $M_{ref}$ consisting of elements indicated by $\{1, 2, \ldots, 25\}$. We put them in a simple row wise increasing sequence, which is to ensure that each element occur only once in the matrix $M_{ref}$. Now corresponding to each MOLS $L_1$ and $L_2$, we  construct a parallel class, as per the Construction \ref{cons:MOLS-RBD}, thereby forming blocks of $P_1$ and $P_2$ by picking elements from $M_{ref}$.
\begin{align*}
P_1 &= \left\{ (2,6,15,19,23),(3,7,11,20,24),(4,8,12,16,25),(5,9,13,17,21),(10,14,18,22,1) \right\}, \\   
P_2 &= \left\{ (2,10,13,16,24),(3,6,14,17,25),(4,7,15,18,21),(5,8,11,19,22),(1,9,12,20,23)\right\}.
 \end{align*}
 The remaining two parallel classes will be constructed using horizontal and vertical elements of $M_{ref}$ as follows:
 \begin{align*}
 P_0 &= \left\{(1,2,3,4,5),(6,7,8,9,10),(11,12,13,14,15),(16,17,18,19,20),(21,22,23,24,25)\right\}, \\
 P_\infty &= \left\{(1,6,11,16,21),(2,7,12,17,22),(3,8,13,18,23),(4,9,14,19,24),(5,10,15,20,25)\right\}.
 \end{align*}
\end{example}    
Let us now consider the construction in the other way, i.e., the converse.

\begin{cons}
\label{Cons-RBD-MOLS}
Consider an RBD$(X,A)$, where $|X| = s^2$ and $A$ consist of $w+2$ numbers of parallel classes such that, each block of a parallel class is of a constant size $s$ and any pair of blocks from different parallel classes have exactly one element in common. Let us denote the elements of $X$ by $\{1, 2, \ldots, s^2\}$, parallel classes by $\{P_0, P_\infty, P_1, \ldots, P_w\}$, and the blocks of $P_l$ by $b^l_i$. Since there are $s$ blocks in each parallel class, therefore $P_{l} = \{b^l_1, b^l_2 , \ldots, b^l_s\}$. Let  $s$ distinct symbols for construction of the Latin Squares be denoted by $Y = \{y_1, y_2, \ldots, y_s\}$.  Now construct $w$-MOLS$(s)$ using RBD$(X,A)$ as follows. 

\begin{enumerate}
\item Use $P_0$ and $P_\infty$ to construct a reference matrix $M_{ref}$ having elements from $\{1, 2, \ldots, s^2\}$ in the following manner:
\begin{equation*}
    M_{ref} =
    \begin{bmatrix}
    b^0_1\cap b^\infty_1 & & & b^0_1\cap b^\infty_2 & &\ldots &  & b^0_1\cap b^\infty_s\\
    b^0_2\cap b^\infty_1 & & & b^0_2 \cap b^\infty_2 & &\ldots & & b^0_2 \cap b^\infty_s\\
    \vdots & & & \vdots & & \ldots & & \vdots\\
    b^0_s\cap b^\infty_1 & & & b^0_s\cap b^\infty_2 & & \ldots & & b^0_s\cap b^\infty_s\\
    \end{bmatrix}.
\end{equation*}
\normalsize

Here $M_{ref}$ contains all the elements of $X$ exactly once. As any two blocks from different parallel classes have exactly one element in common, if $(M_{ref})_{ij} = (M_{ref})_{lm}$ then $b^0_i\cap b^\infty_j = b^0_l\cap b^\infty_m$, and that implies all the blocks  $\{b^0_i,b^0_l , b^\infty_j, b^\infty_m \}$ have one element in common.  Here, $b^0_i$ and $b^0_l$  are the blocks in the Parallel class $P_0$. Similarly, $b^\infty_j$ and $b^\infty_m$ are the blocks of the parallel class $P_\infty$. Since blocks in a Parallel class are mutually disjoint, this is not possible, i.e., $(M_{ref})_{ij} \neq (M_{ref})_{lm}$.

\item Corresponding to the parallel class $P_k$, construct the Latin Squares $L_k$  as follows
\begin{equation*} 
 \text {if} \ (M_{ref})_{ij} \in b^k_t \,\,\, \text{then} \,\,(L_{k})_{ij} = y_t. 
\end{equation*}
 That is, we are substituting $y_t$, wherever the element of the block $b^k_t$ is appearing in $M_{ref}$ to construct $L_k$.  
 Note that $L_k$ is a Latin Square. Since $X = b^k_1 \cup b^k_2 \ldots b^k_{s}$, for every $(M_{ref})_{ij}$ there will be one $b^k_t$ such that  $(M_{ref})_{ij} \in b^k_t$ and since $b^k_i\cap b^k_j = \phi $, $i, j \in \{1, 2, \ldots, s\}$, for each  $(M_{ref})_{ij}$, there will be a unique $b^k_t$ such that $(M_{ref})_{ij} \in b^k_t$. Now if $L_k$ is not a Latin square, then there would be  at least a pair of $(i_1, i_2)$ corresponding to which  $(L_k)_{i_1j} = (L_k)_{i_2j}$ or a pair of $(j_1, j_2)$ corresponding to which $(L_k)_{ij_1} = (L_k)_{ij_j}$. Consider $(L_k)_{i_1j} = (L_k)_{i_2j}$. This implies if $x_1  = (M_{ref})_{i_1j} = b^0_{i_1}\cap b^\infty_j  \in b^k_t$ and $x_2  = (M_{ref})_{i_2j} = b^0_{i_2}\cap b^\infty_j \in b^k_t$.  Thus $x_1$ and $x_2 \in b^\infty_j $. Hence $|b^k_t \cap b^\infty_j | \geq |\{x_1, x_2\}| = 2 $ as $x_1  = (M_{ref})_{i_1j} \neq  (M_{ref})_{i_2j} = x_2 $. This is a contradiction as there is exactly one element common between the blocks of different parallel classes, here $P^\infty$ and $P^k$. Similarly it can be argued that $(L_k)_{ij_1} \neq (L_k)_{ij_j}$ for any pair of $(j_1,j_2)$.
 
\item Repeat the above step for each of the parallel class $P_k, \ k = 1, 2, \ldots, w$, thereby constructing the set of $w$ Latin squares viz $\{L_1, L_2, \ldots, L_w\}$.
\end{enumerate}
\end{cons}
The converse of Lemma \ref{MOLS-RBD} is as follows.
\begin{lemma}
\label{RBD-MOLS}
Given an RBD$(X, A)$, where $|X| = s^2$ and consisting of $w+2$ many parallel classes such that, each block of a parallel class is of a constant size $s$ and any pair of blocks from different parallel classes have exactly one element in common. Then RBD$(X,A)$ can be used to construct $w$-MOLS$(s)$.
\end{lemma}
\begin{proof}
We claim that the set of Latin squares $L_1, L_2, \ldots, L_w$ as constructed above are Mutually Orthogonal Latin Squares of order $s$.

Recalling that a pair of Latin Squares $L_1$ and $L_2$ of same order and constructed from the entries from the same set $Y$ is called mutually orthogonal, if the ordered pair $ ((L_{1})_{ij}, (L_{2})_{ij}) \in \{(Y,Y) \}$ appears exactly once.

Consider the ordered pair, $((L_{k})_{ij}, (L_{m})_{ij}) $. Assume that, $L_k$ and $L_m$ are not Mutually Orthogonal Latin Squares, which implies that, there would be at least one pair of Indices $\{(i,j), (u,v): (i,j) \neq (u,v)\}$ such that, $((L_{k})_{ij}, (L_{m})_{ij}) = ((L_{k})_{uv}, (L_{m})_{uv}) = (y_p, y_q)$. Let $y_p = (L_{k})_{ij} \in b^k_p$ and $y_q = (L_{m})_{ij} \in b^m_q$. This implies $(M_{ref})_{ij}\in b^k_p$ and $b^m_q$. Similarly, $y_p = (L_{k})_{uv} \in b^k_p$ and $y_q = (L_{m})_{uv} \in b^m_q$ which implies $(M_{ref})_{uv} \in b^k_p$ and $b^m_q$. However, then we will have, $b^k_p \cap b^m_q = \{(M_{ref})_{ij},(M_{ref})_{uv}\}$, but if $(i,j) \neq (u,v) $ then $ (M_{ref})_{ij} \neq (M_{ref})_{uv}$. Thus, $|b^k_p \cap b^m_q| \geq 2$ which contradicts the fact that $|b^k_p \cap b^m_q| = 1$. Hence we conclude that $L_k$ and $L_m$ are Mutually Orthogonal Latin Squares.
\end{proof}

We now sketch the above method with an example to obtain a $2$-MOLS$(5)$ from $4$ parallel classes.
\begin{example}
If we proceed with the same set of $4$ parallel classes obtained as in Example \ref{Example-Lemma3}, it would naturally result into the same pair of MOLS$(5)$, i.e., $L_1$ and $L_2$, with which we have started Example \ref{Example-Lemma3}. Therefore, we consider a different set of $4$ parallel classes as follows:
 \begin{align*}
 P_0 &= \left\{(1,2,3,4,5),(6,7,8,9,10),(11,12,13,14,15),(16,17,18,19,20),(21,22,23,24,25)\right\}, \\
 P_\infty &= \left\{(1,6,11,16,21),(2,7,12,17,22),(3,8,13,18,23),(4,9,14,19,24),(5,10,15,20,25) \right\}, \\
 P_1 &= \left\{(1,9,12,20,23),(2,10,13,16,24),(3,6,14,17,25),(4,7,15,18,21),(5,8,11,19,22)\right\},  \\  
 P_2 &= \left\{ (1,7,13,19,25), (2,8,14,20,21), (3,9,15,16,22), (4,10,11,17,23),(5,6,12,18,24) \right\}.
\end{align*}

Note that, $P_0$ and $P_\infty$ are taken as above for convenience, providing the $M_{ref}$ with elements from $X = \{1, 2, \ldots, 25\}$. Observe that any pair of blocks from different parallel classes have exactly one element in common. 
Now corresponding to $s=5$, we simply use  $Y = \{1, 2, 3, 4, 5\}$ as five symbols to construct the Latin square. The remaining parallel classes $P_1$ and $P_2$ are used to construct $L_1$ and $L_2$ respectively, which are the required $2$-MOLS$(5)$. Following the Construction \ref{Cons-RBD-MOLS} above, we obtain Orthogonal Latin Squares $L_1$ and $L_2$ as follows:
 
 {\tiny   
 \begin{equation*}
    M_{ref} =
    \begin{bmatrix}
    1 & 2 & 3 & 4 & 5\\
    6 & 7 & 8 & 9 & 10\\
    11 & 12 & 13 & 14 & 15\\
    16 & 17 & 18 & 19 & 20\\
    21 & 22 & 23 & 24 & 25\\
    \end{bmatrix}, \,\,\,\,
    L_1 =\begin{bmatrix}
    1 & 2 & 3 & 4 & 5\\
    3 & 4 & 5 & 1 & 2\\
    5 & 1 & 2 & 3 & 4\\
    2 & 3 & 4 & 5 & 1\\
    4 & 5 & 1 & 2 & 3\\
    \end{bmatrix},\,\,\,\,
   L_2 = \begin{bmatrix}
    1 & 2 & 3 & 4 & 5\\
    5 & 1 & 2 & 3 & 4\\
    4 & 5 & 1 & 2 & 3\\
    3 & 4 & 5 & 1 & 2\\
    2 & 3 & 4 & 5 & 1\\
    \end{bmatrix}.
    \end{equation*}
    }
\end{example}
Based on this we now proceed for our generic construction idea in the next section.

\section{Definition of APMUBs and general construction ideas}
\label{apmub}
In this section we first present the motivation of proposing such a combinatorial object and then proceed with the general construction ideas.
  
\subsection{Motivation for defining APMUBs and its Characteristics}
\label{Motivation}
Consider a pair of orthonormal bases, say $M_l\text{ and }M_m$. If they are $\beta$-AMUBs then $|\braket{\psi_{i}^{l}|\psi_{j}^{m}}| \leq \frac{\beta}{\sqrt{d}}$, where $\beta$ is bounded by some constant. The definition of $\beta$-AMUBs does not rule out $|\braket{\psi_{i}^{l}|\psi_{j}^{m}}| = 0$.
Let for $\mathfrak{n}_{1}$ pairs, the value of $|\braket{\psi_{i}^{l}|\psi_{j}^{m}}|$ be $0$, and for the remaining pairs $\mathfrak{n}_{2}$, the value of $|\braket{\psi_{i}^{l}|\psi_{j}^{m}}|$ be non-zero ($\neq 0$). If $|\braket{\psi_{i}^{l}|\psi_{j}^{m}}| = \frac{\beta_{ij}}{\sqrt{d}}$ (say), then $\frac{\beta_{ij}}{\sqrt{d}} \leq \frac{\beta}{\sqrt{d}}$. Since $\sum^d_{ij} |\braket{\psi_{i}^{l}|\psi_{j}^{m}}|^2 = d$. Therefore,
$$
\mathfrak{n}_{1}\cdot 0 + \sum_{\substack{i,j=1\\|\braket{\psi_{i}^{l}|\psi_{j}^{m}}|\neq 0}}^{d}\left(\frac{\beta_{ij}}{\sqrt{d}}\right)^{2} = d
\Rightarrow \mathfrak{n}_{2}\frac{\beta_{\min}^{2}}{d} \leq d \leq \mathfrak{n}_{2}\frac{\beta_{\max}^{2}}{d},$$
where, $\beta_{\min} = \min_{ij}\beta_{ij}$ and $\beta_{\max} = \max_{ij}\beta_{ij}$ This implies,
$$\mathfrak{n}_{2}\frac{\beta_{\min}^{2}}{d^2} \leq 1 \leq \mathfrak{n}_{2}\frac{\beta_{\max}^{2}}{d^2}.$$
Note that, $\mathfrak{n}_{1}+\mathfrak{n}_{2} = d^{2} \Rightarrow \frac{\mathfrak{n}_{1}}{d^2} = 1 - \frac{\mathfrak{n}_{2}}{d^2}$, then we have,
\begin{equation}
1 - \frac{1}{\beta_{\min}^{2}} \leq \frac{\mathfrak{n}_{1}}{d^2} \leq 1 - \frac{1}{\beta_{\max}^{2}}.
\end{equation}
Note that, $\frac{\mathfrak{n}_{1}}{d^2}$ is the probability of randomly selecting two orthogonal vectors from two different bases. Therefore, if $\beta_{\min} = \beta = \beta_{\max}$, i.e., $\Delta = \{0,\frac{\beta}{\sqrt{d}}\}$ and $\beta = 1 + \mathcal{O}(d^{-\lambda}), \\lambda>0$, then, $\frac{\mathfrak{n}_{1}}{d^{2}} = \mathcal{O}(d^{-\lambda}), \ \lambda>0$. Hence, with these conditions on $\Delta$ and $\beta$, the probability of any pair randomly selected vectors from different orthonormal bases being orthogonal, tends to $0$. Similarly, the probability that the angle between them is $\frac{\beta}{\sqrt{d}}$ tends to $1$ asymptotically. Since $\beta \rightarrow 1$, as $d$ increases, the bases of APMUBs would behave like MUBs in this sense.

Further note that, the vectors from a set of MUBs form a set of Bi-angular vectors as $\Delta = \{0,\frac{1}{\sqrt{d}}\}$. Now if we restrict $\Delta = \{0,\frac{\beta}{\sqrt{d}}\}$, then basis vectors of AMUBs would form set of Bi-angular vectors. Thus analysis of such AMUBs would also shed light on the study of Bi-angular vectors that has close connections with Weighing Matrices, Error Correcting Codes, Orthogonal spreads, Frame theory, Association Schemes etc.~\cite{Calderbank1997,Best2013,Best2015,Kharaghani2015,Magsino2019,Cahill2018,Casazza2019,Holzmann2010,Kharaghani2018}. 

With this motivation, let us define APMUBs, which is the main focus of this paper.
\begin{definition}
\label{def:APMUB}
The set $\mathbb{M} = \{M_1,M_2,\ldots,M_r\}$ will be called {Almost Perfect MUBs (APMUBs)} if $\Delta  = \left \{0,\frac{\beta}{\sqrt{d}}  \right \}$, i.e., the set 
contains just two values, such that  $\beta = 1+\mathcal{O}(d^{-\lambda}) \leq 2$, $\lambda > 0$. When the bases are real, we call them {Almost Perfect Real MUBs (APRMUBs)}. 
\end{definition}

If the vectors are understood as states of a quantum system then the absolute value of an inner product essentially indicates the overlap between these states. Hence randomly picking two quantum states corresponding to different basis of above set of APMUBs, the overlap will have magnitude equal to $\frac{\beta}{\sqrt{d}}$ with probability almost 1. With negligible probability it will be $0$ which corresponds to the quantum sates being orthogonal.

The sets of basis vectors of APMUB are Bi-angular as they are either orthogonal or have constant absolute value of inner product, i.e., $ \Delta = \{0, \frac{\beta}{\sqrt{d}} \}$. Hence APMUBs form set of Bi-angular vectors with 0 being one of the value of inner product. Thus, for any pair of APMUBs in $\mathbb{C}^d ( \text{or}~ \mathbb{R}^d)$ we have the following lemma.
 
 \begin{lemma}
Consider any pair of APMUBs in $\mathbb{C}^d$ with $ \Delta = \{0, \frac{\beta}{\sqrt{d}} \}$. Then any basis vector of an APMUB will be orthogonal to $(1- \frac{1}{\beta^2} ) \times d$ many basis vectors of another APMUB and will be at angle $\frac{\beta}{\sqrt{d}}$ with remaining $\frac{d}{\beta^2}$ many basis vectors.
\end{lemma}
\begin{proof}
Let $M_l$ and $M_n$ be any pair of APMUBs over $\mathbb{C}^d$ and let $\{\ket{\psi_{i}^{l}}:i=1,2,\ldots,d\}$ and $\{\ket{\psi_{j}^{m}}:j=1,2,\ldots,d\}$ be the corresponding basis vectors. Expressing $\ket{\psi_{i}^{l}}$ as a linear combination of $\{\ket{\psi_{j}^{m}}:j=1,2,\ldots,d\}$, we get,
$$
\ket{\psi_{i}^{l}} = \alpha_{i1}\ket{\psi_{1}^{m}} + \alpha_{i2}\ket{\psi_{2}^{m}} + \ldots + \alpha_{id}\ket{\psi_{d}^{m}},
$$
where $\alpha_{ij} = \braket{\psi_{i}^{l}|\psi_{j}^{m}}$. Since the bases consist of unit vectors, i.e., $\braket{\psi_{i}^{l}|\psi_{i}^{l}} = 1 
\forall l,i$ hence,
$$
\braket{\psi_{i}^{l}|\psi_{i}^{l}} = |\alpha_{i1}|^{2} + |\alpha_{i2}|^{2} + \ldots + |\alpha_{id}|^{2} = \sum_{j=1}^{d}|\braket{\psi_{i}^{l}|\psi_{j}^{m}}|^{2} = 1.
$$
Since $\Delta = \{0,\frac{\beta}{\sqrt{d}}\}$, let us assume that $\ket{\psi_{i}^{l}}$ is orthogonal, i.e., $|\braket{\psi_{i}^{l}|\psi_{j}^{m}}|= 0$  with $\mathfrak{t}_{1}$ many basis vectors of $M^m$, and  make an angle of $\frac{\beta}{\sqrt{d}}$ with remaining $\mathfrak{t}_{2} = d - \mathfrak{t}_{1}$ many basis vectors of $M^m$.  Hence we have, $$\sum_{j=1}^{d}|\braket{\psi_{i}^{l}|\psi_{j}^{m}}|^{2} = 1\Rightarrow \mathfrak{t}_{1} \cdot 0 + \mathfrak{t}_{2} \times \left(\frac{\beta}{\sqrt{d}}\right)^{2} = 1 \ \Rightarrow \ \mathfrak{t}_{1} = \left( 1 - \frac{1}{\beta^2}\right) \times d \text{ and } \mathfrak{t}_{2} = \frac{d}{\beta^2}.$$

Since $\ket{\psi_{i}^{l}}$ is arbitrary basis vector of $M^l$, hence the result. 
\end{proof}

The above result tells that $\beta^2$ can only be a rational number. Further, we have the following corollary considering all the vectors in the set of APMUBs.
\begin{corollary}
Consider a set of $r$ many APMUBs on dimension $d$ with $\Delta = \{0, \frac{\beta}{\sqrt{d}} \}$, that will produce $r \times  d$ vectors. In this set, each vector will have $(d-1) + (r-1) (1- \frac{1}{\beta^2}) d$ many vectors as its orthogonal and $(r-1) \frac{d}{\beta^2}$ many vectors having the dot product $\frac{\beta}{\sqrt{d}}$.
\end{corollary}

The main contribution here is to show that one can construct $\mathcal{O}(\sqrt{d})$ many APMUBs with values of $\beta$ slightly more than one. This says that our method can provide $\mathcal{O}(d^{\frac{3}{2}})$ many Bi-angular vectors where the dot product values are 0 and $\frac{\beta}{\sqrt{d}}$. While there are constructions of $\mathcal{O}(d^2)$ Bi-angular vectors \cite{mikhail2021biangular}, but the angles we obtain with our methods are quite high than the existing constructions. Further we achieve large sparsity and non-zero components of equal magnitude. This is related to coherence property of unit norm vectors. This shows that the construction of APMUBs may produce interesting results in related domain. Further note that any set of Orthonormal Basis vectors always form Unit Norm Tight Frames (UNTF). For a brief introduction on Frame theory and its application in Hilbert space one may refer to \cite{casazza2016brief,casazza2013introduction}. Thus the basis vectors of the set of APMUBs also constitute a Unit Norm Tight Frames apart from being Bi-angular, whereas in general the Bi-angular vectors constructed in \cite{mikhail2021biangular} do not constitute tight frame. To see this, note that since APMUBs are orthonormal basis vectors, hence any arbitrary vector $\ket{u}$ can be uniquely expressed in terms of each of the APMUBs. Thus in this context also the construction of APMUBs may be of independent interest. 

\subsubsection{Connecting with Mutually Unbiased Weighing Matrices}
Let us now demonstrate how the construction of APMUBs bears implications to the existence of Mutually Unbiased Weighing Matrices (MUWM).
\begin{definition}
Let $W_1$ and $W_2$ be a pair of weighing matrices of order $d$ and weight $w$. If $W_1^\dag W_2$ is again a weighing matrix with order $d$ and weight $w$, then the pair is called mutually unbiased weighing matrices (MUWM). Moreover, let $W=\left\{W_1, W_2, \ldots, W_r \right\}$ be a set of weighing matrices such that every pair is mutually unbiased. Then $W$ is referred to as a set of mutually unbiased weighing matrices. If $W$ consists solely of real weighing matrices, it is called a set of mutually unbiased real weighing matrices (MURWM).
\end{definition}

Note that, the MUWMs generalize mutually unbiased Hadamard matrices. The study of mutually unbiased weighing matrices of small orders has been conducted in \cite{Best2015}, where computer searches and some analytical methods were predominantly employed. Moreover, in \cite{harada2015binary}, mutually unbiased real weighing matrices have been used to study the binary codes. In this regard, we like to underline the following technical result.

\begin{lemma}
The existences of the following combinatorial objects are equivalent:
\begin{enumerate}
\item $r$ many APMUBs with $\Delta = \left\{0, \frac{\beta}{\sqrt{d}}\right\}$, and
\item $(r - 1)$ many mutually unbiased weighing matrices of order $d$ and weight $\frac{d}{\beta^2}$.
\end{enumerate}
\end{lemma}
\begin{proof}
$(1) \Rightarrow (2)$: Let $\left\{M_1, M_2,\ldots, M_r \right\}$ be a set of $r$ many APMUB. Choose any weighing matrix, say $M_1$ and consider the set $\left\{M_1^\dag M_1, M_1^\dag  M_2, \ldots, M_1^\dag M_r \right\} = \left\{I, W_2, W_3, \ldots, W_r \right\}$. Since $M_i$'s are unitary matrices, $W_i$ are also unitary. Moreover, since $M_1$ and $M_i$ are APMUB, the elements of $W_i = M_1^\dag M_i $ are from the set $\left\{0, \frac{\beta \exp(i \theta)}{\sqrt{d}}\right\}$ where $\theta \in \mathbb{R}$. Hence, $W_i$ is a weighing matrices with weight $\frac{d}{\beta^2}$.
 
Further, $W_i^\dag W_j = (M_1^\dag  M_i)^\dag (M_1^\dag  M_j) =M_i^\dag M_1 M_1^\dag M_j = M_i^\dag M_j$. Since $M_i$ and $M_j$ are APMUB, the elements of $M_i^\dag M_j$ are from the set $\left\{0, \frac{\beta \exp(i\theta)}{\sqrt{d}}\right\}$ where $\theta \in \mathbb{R}$, making $W_i^\dag W_j$ a weighing matrices with weight $\frac{d}{\beta^2}$. Hence, $W_i$ and $W_j$ are mutually unbiased weighing matrices, for any pair of $W_i,W_j$. Thus $\left\{W_2,W_3,\ldots, W_r\right\}$ is a set of $(r-1)$ mutually unbiased weighing matrices of weight $\frac{d}{\beta^2}$.

$(2)\Rightarrow (1)$ The set of $(r-1)$ mutually unbiased weighing matrices along with identity matrix ($I$), constitute the set of $r$ many APMUB.
\end{proof}
 
In the lemma above, when we restrict the matrices to APRMUB, we obtain a set of mutually unbiased real weighing matrices (MURWM). It's also noteworthy that this lemma parallels the connection between MUBs and mutually unbiased Hadamard matrices (MUHM), where $r$ MUBs are equivalent to $(r-1)$ MUHMs and vice versa.

\subsection{Our general construction ideas}
\label{TheoreticalAnalysis}
Let us now refer to the construction method of orthonormal bases using RBDs as given in~\cite[Section 3]{ak22}. The construction idea of~\cite{ak22} is generic in nature, where unitary matrices can be employed to construct orthogonal bases corresponding to each parallel class of an RBD. However, here we will confine ourselves to the choice of Hadamard Matrices (which are special kind of unitary matrices) for constructing orthonormal bases from parallel classes of an RBD. This will enable us to bound $\beta$ as per~\cite[Theorem 1]{ak22}, which is required for constructing APMUBs with good parameters. The order of the Hadamard matrices used in this construction must be same as the block size of each parallel class. With this method in place for constructing the set of orthonormal bases from RBD, our work here primarily focuses on constructing suitable RBDs, and consequently analyzing the parameters $\Delta, \beta$ and $\epsilon$ of corresponding AMUBs constructed from them.

We focus on constructing RBDs, having constant block size. The constant block size is essential if we want to use Hadamard matrices of same order (real or complex) for all the blocks of the RBD. This is required as it renders all the components of the basis finally constructed to be either zero or of a constant magnitude $(\frac{1}{\sqrt{k}})$, which is the normalizing factor of each basis vector. We will examine when they can satisfy the conditions of APMUBs (APRMUBs). Our construction method results into vectors which are vary sparse with the non-zero components of constant magnitude. This provides large sets of both real as well as complex AMUBs for the dimensions where it is known that not more that two or three real MUBs exist.

We first present a generic result, which is dependent on the existence of suitable RBDs. Thereafter, we explore the methods to construct such RBDs. We further show that if an RBD satisfy $\mu =1$, then it will result into an APMUB. In fact, $\mu$ is the most critical parameter which we control in the construction of RBDs. In all the constructions of AMUBs, the number of elements in an RBD, i.e., $|X|$ can be increased without bound whereas, the parameter $\mu$ remains constant. All our constructions will have this property, which justifies asymptotic analysis of the parameters for AMUBs thus constructed. 

\begin{theorem}
\label{th1}
Consider an RBD$(X, A)$ with $|X|= d= (q-e)(q+f)$, with $q, e, f \in \mathbb{N}$. If the said RBD$(X, A)$ consists of $r$ parallel classes, each having blocks of size $(q-e)$, and if $0 \leq (e+f) \leq \left( \frac{c^2 - \mu^2}{\mu c} \right) d^{\frac{1}{2}}$ where $c$ is some constant, then one can construct $r$ many $\beta$-AMUBs in dimension $d$, where $\beta = \mu \sqrt{\frac{q+f}{q-e}} \leq c$ and $\epsilon = 1 - \frac{1}{q+f}$. When $\mu=1$ and $c = 2$, we will get $r$ many APMUBs with $\beta = 1 + \mathcal{O}(d^{-\frac{1}{2}} ) \leq 2$ and  $\Delta =\{0, \frac{1}{q-e} \}$. Further, if there exists a real Hadamard matrix of order $(q-e)$,  we can construct $r$ many APRMUBs with the same parameters.
\end{theorem}
\begin{proof}
We have $|X|= (q-e)(q+f)$ with each parallel class having the block size $k = (q-e)$. Here $\mu$ is the maximum number of
elements that are common between two blocks from different parallel classes. Following~\cite[ Theorem 1]{ak22}, this implies that, $\beta =  \frac{\mu \sqrt{d}}{q-e} = \mu \sqrt{\frac{q+f}{q-e}}$. Further, using the relation $d= q^2 +(f-e) q - ef$, we obtain  $\beta = \mu(\sqrt{1+x^{2}} + x)$, where $x = \frac{e+f}{2\sqrt{d}}$. From the definition of $\beta$-AMUBs, $\beta$ must be bounded for all values of $d$. Let this bound for $\beta$ be $c$, then $\mu(\sqrt{1+x^{2}} + x) \leq c \Rightarrow 0 \leq (e+f) \leq \left( \frac{c^2 - \mu^2}{\mu c} \right) d^{\frac{1}{2}} $. This inequality can be restated in terms of $q$ as $0 \leq ( c^2 e+\mu^2 f) \leq (c^2-\mu^2) q$, which is the condition for $\beta$ being bounded above by the constant $c$.

In order to see the asymptotic variation of $\beta$ in terms of $q$, we consider the expansion of terms as follows:
\begin{equation}
\label{beta-q}
\beta = \mu\left(1+ \frac{ e+f }{2q} +\frac{ (e+f)(3e-f) }{ 2^3 q^2} +\frac{(e+f)(5e^2-2ef+f^2) }{  2^4 q^3 }+ \ldots \right).
\end{equation} 
To understand the asymptotic variation of $\beta$ in terms of $d$, we again use the relation $d= q^2 +(f-e) q - ef$ to express $q$ in terms of $\sqrt{d}$ and thereafter, expanding the expression for $\beta= \mu \sqrt{\frac{q+f}{q-e}} = \mu(\sqrt{1+x^{2}} + x)$, where $x = \frac{f+e}{2\sqrt{d}}$, in terms of negative power of $\sqrt{d}$ and we obtain
\begin{equation}
\label{beta-d}
\beta  = \mu \left( 1 + \frac{e+f}{2 \sqrt{d}}+ \frac{(e+f)^2}{2^3 d} - \frac{(e+f)^4}{2^7 d^2}+  \frac{(e+f)^6}{2^{10} d^3} - \frac{5 (e+f)^8}{2^{15} d^4} + \ldots \right).
\end{equation}
Thus, for a given $e$ and $f$, for large $d~(\text{or} ~q)$, asymptotically   $\beta = \mu +\mathcal{O}(\frac{1}{q}) = \mu +\mathcal{O}(\frac{1}{\sqrt{d}}) $. Therefore, if $\mu = 1$, the construction yields APMUBs, provided $\beta \leq 2 \Rightarrow c =2$ as we have seen above that $\beta$ is bounded  by $c$. And in this situation we get $0 \leq (e+f) \leq  \frac{3}{2} d^{\frac{1}{2}}$. 

To get the values of the set $\Delta$, note that when $\mu =1$, there is maximum one element common between any pair of blocks from different parallel classes. And since Hadamard matrices are used for constructing orthonormal bases, thus $|\braket{u|v}| = \frac{1}{q-e} $ corresponding to the situations when one element is common between the pair of blocks and $|\braket{u|v}| = 0 $ corresponding to situation, when no elements are common between the pair of blocks. Thus $\Delta = \{0, \frac{1}{q-e} \}$. 

To calculate sparsity, note that for each vector constructed from a block of size $k$, we will have exactly $k$ many non-zero and $d- k$ many zero entries, hence $$\epsilon = \frac{d-k}{d} = 1 - \frac{k}{d}  = 1 - \frac{q-e}{(q+f)(q-e)}= 1 - \frac{1}{q+f}.$$

If a real Hadamard matrix of order $(q-e)$ exists, we can exploit it to obtain $r$ many real approximate MUBs in $\mathbb{R}^d$, with same values of the parameters $\beta$, $\Delta$ and $\epsilon$.
\end{proof}

If the construction in~\cite{ak22} is to be used for APMUBs, then $\mu$ should be 1. Thus $\mu$, which is the maximum number of elements common between any pair of blocks from different parallel classes, is the most critical parameter here. Further note that, $\mu$ is always greater than or equal to $1$, hence, a very limited kinds of RBDs can be used to construct APMUBs. As per Lemma~\ref{RBD1}, $\mu \geq \lceil \frac{k}{s} \rceil$. Thus, for $\mu = 1$, an RBD having a constant block size must have $k \leq s$, i.e., the block size must not be greater than the number of blocks in the parallel class. In this connection, we have noted that an RBD constructed using MOLS have $\mu = 1$. In fact, between any pair of blocks from different parallel classes, in such an RBD, there is exactly one element in common.

In our above theorem, we have $|X| = d = (q-e)(q+f)$, where number of elements in a block is $(q-e)$, i.e., $k = (q-e)$, and number of blocks in a parallel class is $(q+f)$, hence $s= (q+f)$. The reason we are expressing it like this will be clear in Theorem \ref{th:(q-e)(q+f)} where we  demonstrate the construction of such an RBD. Since $e$ and $ f$ are bounded by a positive integer, if $e\geq f$, it will ensure that $k \leq s$. Since$\beta = \mu \sqrt{\frac{s}{k}} = \mu \sqrt{\frac{q+f}{q-e}} $, hence for large $d= (q-e)(q+f)$ we obtain $\beta \rightarrow \mu$, which is also evident from the asymptotic expansion of $\beta $ above. 

For $|X| = d = (q-e)(q+f)$, we can have an RBD, where the block size is $(q+f)$, hence having $(q-e)$ blocks in each parallel class. However, in such a situation, $\mu > 1$ Lemma \ref{RBD1}, and hence we cannot get APMUBs. However, they can provide AMUBs as in~\cite{ak22}. The result of~\cite[Theorem 4]{ak22} is a particular case of this situation, with $e = 0, f = 1$ and $\mu = 2$. In this case, $q+1$ many parallel classes are there in an RBD, each having a constant block size of $(q+1)$. In this case, $\beta = 2 \sqrt{\frac{q}{q+1}} = 2 - \mathcal{O}(\frac{1}{\sqrt{d}})$, i.e., though the maximum value of the inner product was slightly less than $\frac{2}{\sqrt{d}}$, but asymptotically $\beta$ converges to 2.  Further $\Delta$ is also not two-valued. Thus the construction did not satisfy  the conditions needed for APMUBs which is $\beta = 1+ \mathcal{O}(d^{-\lambda})$, for some $\lambda > 0 $ and $\Delta$ being the set consisting of just two elements with one being 0. 

Thus, in order to obtain APMUBs, the RBDs in use must have $\mu=1$ and all the block sizes must be same. With this understanding, we will explore more suitable designs in the following sections, so that the upper bound on the absolute inner product values can be improved than the results presented in~\cite{ak22} to obtain APMUBs.  

\section{Exact constructions of APMUBs through RBDs}
\label{ConsAPMUB}
As followed from previous section, an RBD having constant block size must have $\mu =1$, in order to obtain APMUBs. In this section we explain the constructions of such RBDs. Since our focus is to build APMUBs in composite, we concentrate on $d= k \cdot s$, and consider two categories. 
\begin{itemize}
\item The first  construction, being generic in nature, will work for  any composite $d = k\cdot s = (s-e)s $ with $0 \leq e \leq \frac{3}{2}d^{\frac{1}{2}}$. Here the number of APMUBs is  at least $N(s)+1$ when $e>0$ and $N(s)+2$ when $e=0$. 
\item The second one is considered when $d$ can be expressed as $(q-e)(q+f),\,\, 0< \ f \leq e$, where $q$ is some power of prime. Here the number of APMUBs is at least $\lfloor \frac{q-e}{f} \rfloor+1$.
\end{itemize}

These are presented in the Sections~\ref{case1} and ~\ref{case2} respectively as below.
Thus here our approach to is to obtain large numbers of APMUBs, if the composite dimension $d$ can be expressed in some generic form. In the first category our starting point is $w$-MOLS$(s)$ and in the second category we initiate with $(q^2,q,1)$-ARBIBD. In each case, we will first demonstrate the construction with an example, then outline the algorithm for the respective construction and then provide the proof of correctness. 

\subsection{$d = k \cdot s= (s-e)s$, $0 \leq e \leq \frac{3}{2}d^{\frac{1}{2}}$}
\label{case1}
Let us first demonstrate the method by explicitly constructing RBD$(X,A)$ with $|X|=  2 \cdot 5 = 10$, i.e., here $s=5$ and $k=2$.  

\begin{enumerate}
\item To begin with, consider the following $4$-MOLS$(5)$ and the $M_{ref}$:

{\tiny
$LS_1 = \begin{bmatrix}
     1& 2   &3  &4  &5   \\
      5& 1   &2 &3  &4  \\
     4& 5  &1  &2  &3 \\
     3& 4 &5  &1   &2 \\
  2& 3  &4  &5  &1  \\
 \end{bmatrix},
 LS_2 = \begin{bmatrix}
     1& 2   &3  &4  &5   \\
      4& 5   &1  &2  &3  \\
     2& 3   &4  &5  &1 \\
     5& 1 &2  &3   &4 \\
  3& 4   &5  &1  &2  \\
 \end{bmatrix},
 LS_3 = \begin{bmatrix}
     1& 2   &3  &4  &5   \\
     3& 4   &5  &1  &2  \\
     5& 1   &2  &3  &4 \\
     2& 3 &4  &5   &1 \\
  4& 5   &1  &2  &3  \\
 \end{bmatrix}$,

$LS_4 = \begin{bmatrix}
     1& 2   &3  &4  &5   \\
      2& 3   &4  &5  &1  \\
     3& 4   &5  &1  &2 \\
     4& 5 &1  &2   &3 \\
  5& 1   &2  &3  &4  \\
 \end{bmatrix},
 M_{ref} = \begin{bmatrix}
   1& 2& 3&4&5 \\
  6&7 &8& 9 &10  \\
 11&12  &13&14 &15 \\
 16 &17  &18 &19 &20 \\
21&22 &23 &24 &25 \\
 \end{bmatrix}$.
}

\item Using Construction \ref{cons:MOLS-RBD}, we construct RBD$(\bar{X},\bar{A})$, with $|\bar{X}| = 25$ and $\bar{A}$ having $6$ parallel classes. We used $M_{ref}$ as shown above to obtain the following RBD. We collect blocks from each parallel class, in a $5 \times 5$ matrix, where each row represents one block of the parallel class, and index the blocks as $\bar{b}^l_i$, where $l$ represents the index of parallel class and $i$ represents the block number within the parallel class.

{\tiny
$\bar{P}_1 = \begin{bmatrix}
     \bar{b}^1_5=  \{ 1& 7   &13  &19  &25 \}  \\
     \bar{b}^1_4=  \{ 2& 8   &14  &20  &21 \}  \\
     \bar{b}^1_3=  \{ 3& 9   &15  &16  &22 \} \\
     \bar{b}^1_2=  \{ 4& 10 &11  &17   &23 \} \\
     \bar{b}^1_1= \{ 5& 6   &12  &18  &24  \}\\
 \end{bmatrix},
\bar{P}_2 = \begin{bmatrix}
     \bar{b}^2_5= \{1&8  &15  &17  &24 \}\\
     \bar{b}^2_4= \{ 2&9  &11  &18  &25 \} \\
     \bar{b}^2_3= \{3&10 &12 &19  &21\}\\
     \bar{b}^2_2= \{ 4 &6  &13 &20  &22 \}\\
     \bar{b}^2_1= \{5&7   &14 &16  &23\}\\
 \end{bmatrix}$,

$\bar{P}_3 = \begin{bmatrix}
    \bar{b}^3_5= \{ 1&9   &12 &20 &23 \}\\
    \bar{b}^3_4= \{ 2&10 &13 & 16 &24 \}  \\
    \bar{b}^3_3= \{  3&6   &14 &17 &25  \}\\
    \bar{b}^3_2= \{  4&7   &15 &18 &21 \} \\
    \bar{b}^3_1= \{  5&8   &11 &19  &22 \} \\
 \end{bmatrix},
\bar{P}_4 = \begin{bmatrix}
     \bar{b}^4_5=\{ 1&10 &14&18&22\} \\
     \bar{b}^4_4=\{ 2&6  &15 &19&23\}   \\
     \bar{b}^4_3=\{ 3&7  &11&20&24\}  \\
     \bar{b}^4_2=\{ 4 &8 &12 &16 &25 \} \\
     \bar{b}^4_1=\{  5&9 &13 &17 &21 \} \\
 \end{bmatrix}$,

$\bar{P}_{\infty} = \begin{bmatrix}
   \bar{b}^5_5= \{1&6&11&16&21\} \\
   \bar{b}^5_4=\{  2&7 &12 &17 &22  \} \\
   \bar{b}^5_3=\{  3&8  &13&18 &23 \} \\
   \bar{b}^5_2= \{ 4 &9  &14 &19 &24 \} \\
   \bar{b}^5_1=\{    5&10 &15 &20 &25 \}  \\
 \end{bmatrix},
  \bar{P}_0 = \begin{bmatrix}
  \bar{b}^6_5=\{1& 2& 3&4&5 \}\\
  \bar{b}^6_4=\{ 6&7 &8& 9 &10 \}  \\
  \bar{b}^6_3=\{ 11&12  &13&14 &15  \}\\
  \bar{b}^6_2=\{ 16 &17  &18 &19 &20 \} \\
  \bar{b}^6_1=\{ 21&22 &23 &24 &25 \} \\
 \end{bmatrix}
$.
}

Here $\bar{A} = \{\bar{P}_1 \cup \bar{P}_2 \cup \bar{P}_3 \cup \bar{P}_4 \cup \bar{P}_0 \cup \bar{P}_{\infty} \}$. Note that, any pair of blocks from different parallel classes has exactly one element in common. 

\item Now we remove any 3 blocks from the parallel class $\bar{P}_0$, say $\{\bar{b}^6_1, \bar{b}^6_2, \bar{b}^6_3\}$.

\item Next we remove the elements contained in this block from the entire design $(\bar{X}, \bar{A})$.

\item Then we discard $\bar{P}_0$ from the design. Here we get RBD$(X, A)$ consisting of $5$ parallel classes, each having $5$ blocks of size $2$, where $X = \{1, 2, 3, 4, 5, 6, 7, 8, 9, 10\}$ and $A = \{P_1 \cup P_2 \cup P_3 \cup P_4 \cup P_\infty \}$. Explicitly, we have,

\small
\begin{equation*}
P_1 = \begin{bmatrix}
     b^1_5=  \{ 1& 7   \}  \\
     b^1_4=  \{ 2& 8    \}  \\
     b^1_3=  \{ 3& 9    \} \\
     b^1_2=  \{ 4& 10  \} \\
     b^1_1= \{ 5& 6     \}\\
 \end{bmatrix},
P_2 = \begin{bmatrix}
    b^2_5= \{1&8   \}\\
    b^2_4= \{ 2&9   \} \\
    b^2_3= \{3&10 \}\\
    b^2_2= \{ 4 &6   \}\\
    b^2_1= \{5&7   \}\\
 \end{bmatrix},
 P_3= \begin{bmatrix}
     b^3_5= \{ 1&9    \}\\
     b^3_4= \{ 2&10  \}  \\
     b^3_3= \{  3&6  \}\\
     b^3_2= \{  4&7    \} \\
     b^3_1= \{  5&8    \} \\
 \end{bmatrix},
 \end{equation*}
\begin{equation*} 
P_4 = \begin{bmatrix}
    b^4_5=\{ 1&10 \} \\
    b^4_4=\{ 2&6  \}   \\
    b^4_3=\{ 3&7  \}  \\
    b^4_2=\{ 4 &8  \} \\
    b^4_1=\{  5&9  \} \\
 \end{bmatrix},
P_\infty = \begin{bmatrix}
   b^5_5= \{1&6\} \\
   b^5_4=\{  2&7   \} \\
   b^5_3=\{  3&8   \} \\
   b^5_2= \{ 4 &9   \} \\
   b^5_1=\{    5&10  \}  \\
 \end{bmatrix}.
 \end{equation*}
\end{enumerate}
\normalsize

Note that, for this particular case using $(10,2,1)$-BIBD, one can construct an RBD with $9$ parallel classes each having $5$ many blocks of constant block size $2$.  
One must note that this construction may not provide RBDs having maximum number of parallel classes, with constant block size and $\mu = 1$. Even if we include the parallel class $P_0 = \bar{P}_0 \setminus \{\bar{b}^6_1, \bar{b}^6_2, \bar{b}^6_3\} $ in the RBD$(X,A)$, it will not change the value of $\mu$ which will remain equal to 1. However, the block size of $P_0$ will be 5 then and hence RBD(X,A) will contain two different block sizes. Therefore we discard the $P_0$. Nevertheless we will see that even $P_0$ can be used to construct orthonormal basis which will be mutually unbiased with all the orthonormal basis constructed using $\{P_1 \cup P_2 \cup P_3 \cup P_4 \cup P_\infty \}$.

The technique is more formally explained for the general case in Construction \ref{cons:(s-e)s} below. 
\begin{cons}
\label{cons:(s-e)s}
\normalfont
Let $d= k \cdot s =(s-e)s$, with $0< e \leq s$.

\begin{enumerate}
\item Using Construction \ref{cons:MOLS-RBD}, construct $RBD(\bar{X}, \bar{A})$, where $\bar{X} = \{1, 2, \ldots, s^2$\}. It will have $r = N(s)+2$ many parallel classes, namely $\{\bar{P}_1, \bar{P}_2, \ldots, \bar{P}_w, \bar{P}_0, \bar{P}_\infty \}$, each having $s$ many blocks of constant size $s$. Denoting blocks of the parallel class $\bar{P}_l$ with $\bar{b}^l_i$, for $i = 1, 2, \ldots, s$, we note that between any two blocks from different parallel classes, there is exactly one element in common, i.e., $|\bar{b}^l_i \cap \bar{b}^m_j| = 1, \ \forall ~l \neq m$.

\item Pick a parallel class, say $\bar{P}_{0}$. Remove $e$ many blocks from it and denote as $S = \{\bar{b}^0_1 \cup \bar{b}^0_2 \cup \ldots \cup \bar{b}^0_e \}$.

\item Remove the elements in $S$ from $\bar{X}$ and let us denote the new set with $X$, i.e., $X= \bar{X} \setminus S$. Further, we remove the elements in $S$ from the parallel classes $\{\bar{P}_2, \bar{P}_3, \ldots, \bar{P}_{w},\bar{P}_\infty \}$ and denote them by $P_l$, for $l = 2, 3, \ldots, r$, i.e., $P_l =  \bar{P}_{l} \setminus S$. Then $A =\{P_2, P_3, \ldots, P_{w}, P_\infty \}$ with $P_l = \{b^l_1, b^l_2, \ldots, b^l_q \}$, where $b^l_i = \bar{b}^l_i \setminus S$.

\item Discard the parallel class $\bar{P}_0$. The resulting RBD$(X,A)$ is the required design. For convenience, rename the elements from 1 to $(s-e)s$.
\end{enumerate}
\end{cons}

We claim that the above design $(X, A)$ is an RBD, such that $|X| = (s-e)s$ and $A$ consist of $N(s)+1$ many parallel classes, i.e., $A = \{P_2, P_3, \ldots,P_{w}, P_{\infty} \}$, each having $s$ many blocks, i.e., $P_l = \{ b^l_1, b^l_2, \ldots, b^l_s\}, l = 1, 2, \ldots, s$, each of size $(s-e)$, i.e., $ |b^l_i| = (s-e) \ \forall i,l$, such that blocks from different parallel classes have at most one element in common, i.e., $ |b^l_i \cap b^m_j| \leq 1 \ \forall \, l \neq m$. We formalize this in the form of a lemma below.

\begin{lemma}
\label{lemxx1}
Let $d=(s-e)s$ for $s, e \in \mathbb{N}$ with $0< e \leq s$. Then one can construct an RBD$(X,A)$, with $|X| = d$ having constant block size $(s-e)$ with $\mu = 1$, and having $N(s)+1$ many parallel classes, where $N(s)$ is the number of MOLS$(s)$.
\end{lemma}
\begin{proof}
Refer to Construction \ref{cons:(s-e)s} above. In RBD$(\bar{X}, \bar{A})$ any pair of blocks from different parallel classes is of size $s$ and has exactly one element in common, i.e., $ |\bar{b}^l_i \cap \bar{b}^m_j| = 1 \ \forall \ l \neq m$. Hence removal of the elements $S = \{ \bar{b}^1_1 \cup \bar{b}^1_2 \cup \ldots \cup \bar{b}^1_e \}$ from entire design will discard exactly $e$ elements from each block $\bar{b}^l_i, l\neq 1$. Hence, the blocks $b^l_i= \bar{b}^l_i \backslash S$ will be of constant size  $|b^l_i| = s-e$ and $|b^l_i \cap b^m_j| \leq 1 \ \forall \, l \neq m$. 
\end{proof}
Now we can use this RBD$(X,A)$ to construct APMUBs in dimension $d= |X|= (s-e)s$ following Theorem~\ref{th1}. 
\begin{theorem}
\label{th:(s-e)s}
Let $d=(s-e)s$ for $s, e \in \mathbb{N}$ with $0 < e \leq \frac{3}{2} d^{\frac{1}{2}}$ . Then there exist $N(s)+1$ many APMUBs with $\Delta = \{0, \frac{1}{s-e}\}$ and $\beta = \sqrt{\frac{s}{s-e}} = 1 + \mathcal{O}(d^ {-\lambda}) \leq 2$, where $\lambda = \frac{1}{2}$ and sparsity  $\epsilon = 1 - \frac{1}{s}$.  Further, if there exists a real Hadamard matrix of order $(s-e)$, then we can construct  $N(s)+1$ many APRMUBs with the same parameters. For the case $e=0$, there exist $N(s)+2$ many MUBs, and if there exists a real Hadamard matrix of order $s$, then we can construct $N(s)+2$ many Real MUBs.
\end{theorem}
\begin{proof}
In order to show that we can produce such number of APMUBs, let us consider an RBD$(X,A)$ with $|X| = (s-e)s$ having $N(s)+1$ parallel classes of constant block size $(s-e)$ such that between the blocks from different parallel classes, there is at most one element in common, and hence $\mu = 1$. This follows from Lemma~\ref{lemxx1}. Then using this RBD along with a Hadamard matrix of order $(s-e)$, we can construct orthonormal bases following Theorem~\ref{th1}, with the values of $\Delta, \beta, \epsilon$ by substituting $q = s$, and $f = 0$. Further, the condition that $\beta \leq 2$ for APMUB gives $e \leq \frac{3}{2} d^{\frac{1}{2}}$. In terms of $s$ this inequality becomes $e \leq \frac{3}{4} s$ and in terms of $k = (s-e)$, the inequality becomes $s \leq 4k$ so that $\beta \leq 2$. The parameters of APMUBs are  $\Delta = \{0, \frac{1}{s-e} \}$, $\beta = \sqrt{\frac{s}{s-e}} = 1 + \mathcal{O}(d^{-\frac{1}{2}}) \leq 2$ and $\epsilon = 1- \frac{1}{s}$. Since the number of parallel classes is one more than number of Mutually Orthogonal Latin Squares of Order $s$, we have $r = N(s)+1$. 

In the situation when $d = s^2$, i.e., $e=0$, using Construction~\ref{cons:MOLS-RBD} above, we obtain RBD$(X,A)$ having $N(s)+2$ parallel classes, such that any pair of blocks have exactly one point in common.  Now using this RBD$(X,A)$,  along with a Hadamard matrix of order $s$, we can construct orthonormal bases following Theorem~\ref{th1} which will provide $N(s)+2$ many MUBs. Hence for $e=0$, the result follows directly. Construction~\ref{cons:(s-e)s} is applied only when $e> 0$. Thus when $d = s^2$, we get $N(s)+2$ MUBs, and when real Hadamard matrix of order $s$ is available, that can be used to construct $N(s)+2$ real MUBs. 
\end{proof} 
\begin{remark}
Note that we can construct an orthonormal basis corresponding to the parallel class $P_0 = \bar{P}_0 \setminus S$, again having $(s-e)s$ elements. These basis vectors will be mutually unbiased with all the orthonormal bases constructed using the parallel classes $P_1, P_2, \ldots P_w, P_\infty$. However, if we include it in the set of orthonormal bases then $\Delta = \{0, \frac{1}{\sqrt{d}}, \frac{\beta}{\sqrt{d}} \}$ will have three values, and the condition of APMUB will not be satisfied. Thus we ignore $P_0$, even though it provides an orthonormal basis which is mutually unbiased with all the bases constructed above.
\end{remark}
For a composite $s$, $N(s) \rightarrow \mathcal{O}(s^{\frac{1}{14.8}}) \ll s-1$, which is an upper bound~\cite{HandBook,stinson2007combinatorial,wocjan2004new}. Thus in case $d$ can be expressed as $d =k\cdot s = (s-e)s$, where $s, e \in \mathbb{N}$ with $k < s \leq 4k$  (or $0 <  e \leq \frac{3}{4}s$), we can always construct $N(s)+1$ many APMUBs. We refer to this as {\sf{Mutually Orthogonal Latin Square Lower Bound construction}} for APMUBs. For example, given $d= 2^2 \cdot 3^2 \cdot 5 \cdot 7$, 
\begin{itemize}
\item this can be factored as $d= 35 \cdot 36$, which will provide $N(36)+ 1 = 9$ many APMUBs with $\beta =1.01$, 
\item or $d= 30 \cdot 42$ which will give $N(42)+ 1 = 6$ many APMUBs with $\beta = 1.18$, 
\item or $d = 28 \cdot 45$, which will give $N(45)+ 1 = 7$ many APMUB with $\beta = 1.27$. 
\end{itemize}
Here, $N(36) = 8, N(42) = 5$ and $N(45) = 6\,$ are the presently known values of the maximum number of MOLS of these orders \cite{HandBook}. Let us now explain the significance of our construction method through this example.
\begin{itemize}
\item The number of complex MUBs for this $d= 2^2 \cdot 3^2 \cdot 5 \cdot 7$ which can be constructed using prime power decomposition, and then taking tensor product, would be $2^2 +1 = 5$. This is the lower bound and there is no better known result than this in the number of MUBs in this dimension. The value of $\beta$ is $1$ in this case as exact MUBs are referred.
\item Expressing $d= 35 \cdot 36$, we have more number of APMUBs (9 many) than MUBs, with $\beta =1.01$. 
\item Further, expressing $d = 28 \cdot 45$, we can get  7 many APRMUBs with $\beta = 1.27$. This is because, we have real Hadamard matrix on the dimension $4 \cdot 7 = 28$.
\end{itemize}

Our Theorem \ref{th:(s-e)s} can be compared with that of \cite[Theorem 3]{wocjan2004new}, where the result could be achieved using $(k,s)$-nets. This, in turn, can be constructed from Mutually Orthogonal Latin Squares of order $s$. Thus, for a square dimension $d = s^2$, there would be $N(s)+2$ many MUBs. Moreover, the result in \cite[Corollary 3]{ak22} points out that using RBDs, one can construct $q+1$ many MUBs of dimension $d$ when $d = q^2$. Note that $N(q) = q-1$ and thus the number of MUBs from~\cite[Corollary 3]{ak22} is same as that presented in~ \cite[Theorem 3]{wocjan2004new}. The construction of~\cite[Corollary 3]{ak22} had the advantage of using different Hadamard matrices for the construction of each MUBs, whereas a single Hadamard matrix can be used for construction of MUBs in~\cite[Theorem 3]{wocjan2004new}. 

Since Theorem~\ref{th:(s-e)s} is based on RBDs as in~\cite[Corollary 3]{ak22}, here also different Hadamard matrices can be used for the construction of each MUBs. Hence, Theorem \ref{th:(s-e)s} is a generalization of~\cite[Theorem 3]{wocjan2004new} and~\cite[Corollary 3]{ak22} as enumerated below. 
\begin{enumerate}
\item For $d = s^2$, Theorem \ref{th:(s-e)s} reproduces the results of ~\cite[Theorem 3]{wocjan2004new} in terms of number of MUBs constructed.
\item For $d= q^2$, Theorem \ref{th:(s-e)s} reproduces the results of ~\cite[Corollary 3]{ak22}, in terms of  having the advantage of using different Hadamard matrices for the construction of each MUBs and the number of MUBs constructed. 
\item As the additional contribution, for any composite $d = k\cdot s= (s-e)\cdot s$, with $0< e \leq \frac{3}{2} d^{\frac{1}{2}}$, Theorem \ref{th:(s-e)s} provides MOLS$(s)+1$ many APMUBs and one can also use different Hadamard matrices for the construction of each basis.
\end{enumerate}

The case $N(s)= s-1$ corresponds to affine plane of order $s$ and the corresponding RBD is called Affine Resolvable BIBD. When $s = q$, where $q$ is a prime power, we have well known methods to construct Affine Resolvable $(q^2,q,1)$-BIBDs. Hence in such a situation we will have $q \sim \mathcal{O}(\sqrt{d})$ many APMUBs for composite dimensions which are not square. For example, if $d = 3^4 \cdot 7 = 21 \cdot 27$, we have 29 APMUBs with $\beta = 1.13$ or for $d = 2^4\cdot 3 = 6\cdot 8$ we obtain 9 APMUBs, with $\beta = 1.15$ whereas number of MUBs is 8 and 4 respectively in these cases.
 
Now let us explain the consequences for Approximate Real MUBs. When $q$ is a prime power, and real Hadamard matrix of order $(q-e)$ exists, then we will obtain $(q+1) > \sqrt{d}$ many APRMUBs, which provides large numbers of such objects over $\mathbb{R}^d$. This is presented in the following result.

\begin{corollary}
\label{cor:(q-e)q}
Let $d=(q-e)q$, where $q$ is a prime power and $e \in \mathbb{N}$, with $0< e \leq \frac{3}{2} d^{\frac{1}{2}}$. Then there exist $q+1$ many APMUBs with $\Delta = \{0, \frac{1}{q-e} \}$, $\beta = \sqrt{\frac{q}{q-e}} = 1+ \mathcal{O}(d^{-\lambda})$, where $\lambda = \frac{1}{2}$ and $\epsilon = 1- \frac{1}{q}$. Further, if there exist real Hadamard matrices of order $(q-e)$, then one can construct $q+1$ many Almost Perfect Real MUBs with same parameters.
\end{corollary}

The condition $e \leq \frac{3}{2} d^{\frac{1}{2}}$ is because we require $\beta \leq 2$ for APMUBs. If $e > \frac{3}{2} d^{\frac{1}{2}}$ then $\beta > 2$, hence the constructed Approximate MUBs will not satisfy the criteria for APMUBs [Definition \ref{def:APMUB}]. Further examining the expression for $\beta$ in Equation~\ref{beta-d}, for this particular construction with $\mu = 1, f = 0$, we obtain the best possible APRMUBs when $e=1$. That is, we have this situation when the dimensions are of the form $d = (q-1)q$. This we formally state in the following corollary.

\begin{corollary}
\label{cor:(q-1)q}
Consider $d = (q-1)q$, such that $q$ is a prime power and assume that a real Hadamard matrix of order $(q-1)$ exists. Then one can construct $q$ many Almost Perfect Real MUBs in dimension $d$ with $\Delta  =  \left\{0, \frac{1}{q-1} \right\}$, $\beta = \sqrt{\frac{q}{q-1}} = 1+\mathcal{O}(d^{-\frac{1}{2}})$, and $\epsilon = 1 - \frac{1}{q}$. 
\end{corollary}

For example, when $d = 20$ and $156$, there would be respectively $5$ and $13$ APRMUBs of the above type. One may note that we have $\lceil \sqrt{d} \rceil$ many APRMUBs in this case. If $m = \frac{q-3}{2} \equiv 1 \bmod 4$ and $m$ is some prime power, then using the Paley Construction \cite{paley1933orthogonal}, one can obtain Hadamard 
matrix of order $2(m+1) = q-1$. Hence for any prime power $q \equiv 1 \bmod 4$, if $\frac{q-3}{2}$ is also some prime power and is equal  to $1 \bmod 4$, then the real Hadamard matrix of order $q-1$ will necessarily exist through the Paley Construction. For example, one can  consider $q = 13, 29, 5^3$ etc. For such $q$'s, the result will become independent of the Hadamard Conjecture.

\subsection{$d= k \cdot s= (q-e)(q+f), \ 0 < f \leq e $ and $ 0 < (e+f) \leq \frac{3}{2} d^{\frac{1}{2}} $}
\label{case2}
As noted in Corollary \ref{cor:(q-e)q} that if $d$ can be expresses as $k\cdot s =(q-e)q$ with $\beta = \sqrt{\frac{s}{k}} =  \sqrt{\frac{q}{q-e}}\leq 2$, then there we can construct $q =\mathcal{O}(\sqrt{d})$ many APMUBs. However, if $d$ can not be expressed in this form then one can construct $N(s)$ many APMUBs if $d$ can be expressed as $k\cdot s$ with $s$ a composite such that $k\leq s\leq 4k$, i.e., for example in the cases $d=\{ 2\cdot 3^2,  \ 2\cdot 11, \  2\cdot 3\cdot 7, \  2^2\cdot 3\cdot 7, \ldots\} $  etc. The condition $k\leq s\leq 4k$ ensures $\beta \leq 2$ and if $d$ can be expressed as $(s-e)\cdot s$ then it is equivalent to $0 \leq e \leq \frac{3}{2} d^{\frac{1}{2}}$. The best known lower bound for general $s$ is $N(s) \sim s^{\frac{1}{14.8}}$ which is much less than $s$. 

In order to obtain significantly larger number of APMUBs, for the dimensions that cannot be expressed in the form $d= (q-e)q$ where $q$ is some prime power with $0 < e \leq \frac{3}{2} d^{\frac{1}{2}}$, we now consider the form of  $d = (q-e)(q+f)$, such that $q$ is a prime power with $e, f \in \mathbb{N}$. First we show that, in such a case if $e\geq f$, then we can construct RBD$(X,A)$ with $|X| = d$ having constant block size $(q-e)$ such that $A$ can be partitioned into at least $r= \lfloor \frac{q-(e-f)}{f} \rfloor= \lfloor \frac{q-e}{f} \rfloor+1$ many parallel classes. Hence such an RBD$(X,A)$ can be used to construct $\lfloor \frac{q-e}{f}\rfloor + 1$ many orthonormal bases following~\cite[ Theorem 1]{ak22}. Further, if $ 0 < (e+f) \leq \frac{3}{2} d^{\frac{1}{2}} $, then $\beta \leq 2$, and these orthonormal bases would be APMUBs, thus providing us $\mathcal{O}(q)$ many APMUBs in such a situation. 

We explain this construction in two parts. First we consider a $(q^2, q, 1)$ Affine Resolvable BIBD as the input. We call this RBD$(\bar{X}, \bar{A})$ where $|\bar{X}|=q^2$  and all the blocks of $A$ is of the same size $q$ and number of parallel classes in $A$ is $q+1$. We use this to construct RBD$(\widetilde{X},\widetilde{A})$,  where $|\widetilde{X}| = (q-e)(q+f)$ having same number of parallel class $(q+1)$, but  the blocks are not of the same size. The sizes of the blocks are from set $ \{(q-e), (q-e+1) , \ldots, (q-e+f), q\}$. 

In the second part we use the RBD$(\widetilde{X},\widetilde{A})$ as input and construct RBD$(X,A)$ where $|X| = (q-e)(q+f)$ such that each block in $A$ is of size $(q-e)$. However, now the number of parallel class reduces to $ \lfloor \frac{q-e}{f} \rfloor+1$. 
 
To demonstrate the first part of the construction, we take $|X|= (7-3)(7+1) = 4 \cdot 8= 32$,  where $q= 7, e= 3$ and $f= 1$. We use an Affine Resolvable $(7^2,7,1)$-BIBD which we  call  RBD$(\bar{X}, \bar{A})$. It will consist of eight parallel classes. Each parallel class would consist of seven blocks of constant size 7. We represent each parallel class as a $7 \times 7$ matrix, where each row represent one block of the parallel class. Hence there would be eight such matrices as below to represent the design.

{\tiny
\begin{equation*}
\bar{P}_1 = \begin{bmatrix}
\bar{b}^1_7= \{ 1 & 2 & 3 & 4 & 5 & 6 & 7 \} \\
\bar{b}^1_6= \{ 8 & 9 & 10 & 11 & 12 & 13 & 14 \}\\
\bar{b}^1_5= \{ 15 & 16 & 17 & 18 & 19 & 20 & 21 \}\\
\bar{b}^1_4= \{ 22 & 23 & 24 & 25 & 26 & 27 & 28 \}\\
\bar{b}^1_3= \{ 29 & 30 & 31 & 32 & {\color{red}33} & {\color{red}34} & {\color{red}35}  \}\\
\bar{b}^1_2= \{ {\color{red}36} & {\color{red}37} & {\color{red}38} & {\color{red}39} & {\color{red}40} & {\color{red}41} & {\color{red}42} \}\\
\bar{b}^1_1= \{ {\color{red}43} & {\color{red}44} & {\color{red}45} & {\color{red}46} & {\color{red}47} & {\color{red}48} & {\color{red}49} \}\\
\end{bmatrix},
\bar{P}_2 = \begin{bmatrix}
\bar{b}^2_7= \{ 1 & 9 & 17 & 25 & {\color{red}33} & {\color{red}41} & {\color{red}49} \}\\
 \bar{b}^2_6= \{ 2 & 10 & 18 & 26 & {\color{red}34} & {\color{red}42} & {\color{red}43}\} \\
\bar{b}^2_5= \{  3 & 11 & 19 & 27 & {\color{red}35} & {\color{red}36} & {\color{red}44} \}\\
 \bar{b}^2_4= \{ 4 & 12 & 20 & 28 & 29 & {\color{red}37} & {\color{red}45} \}\\
 \bar{b}^2_3= \{ 5 & 13 & 21 & 22 & 30 & {\color{red}38} & {\color{red}46} \}\\
 \bar{b}^2_2= \{ 6 & 14 & 15 & 23 & 31 & {\color{red}39} & {\color{red}47} \}\\
 \bar{b}^2_1= \{ 7 & 8 & 16 & 24 & 32 & {\color{red}40} & {\color{red}48}  \}\\
 \end{bmatrix},
\end{equation*}

\begin{equation*}
\bar{P}_3 = \begin{bmatrix}
\bar{b}^3_7= \{  1 & 10 & 19 & 28 & 30 & {\color{red}39} & {\color{red}48} \}\\
\bar{b}^3_6= \{  2 & 11 & 20 & 22 & 31 & {\color{red}40} & {\color{red}49} \}\\
\bar{b}^3_5= \{  3 & 12 & 21 & 23 & 32 & {\color{red}41} & {\color{red}43} \}\\
\bar{b}^3_4= \{  4 & 13 & 15 & 24 & {\color{red}33} & {\color{red}42} & {\color{red}44} \} \\
\bar{b}^3_3= \{  5 & 14 & 16 & 25 & {\color{red}34} & {\color{red}36} & {\color{red}45}\}\\
 \bar{b}^3_2= \{ 6 & 8 & 17 & 26 & {\color{red}35} & {\color{red}37} & {\color{red}46}  \}\\
\bar{b}^3_1= \{ 7 & 9 & 18 & 27 & 29 & {\color{red}38} & {\color{red}47} \}\\
 \end{bmatrix},
\bar{P}_4 = \begin{bmatrix}
\bar{b}^4_7= \{  1 & 11 & 21 & 24 & {\color{red}34} & {\color{red}37} & {\color{red}47} \}\\
\bar{b}^4_6= \{  2 & 12 & 15 & 25 & {\color{red}35} & {\color{red}38} & {\color{red} 48} \}\\
\bar{b}^4_5= \{  3 & 13 & 16 & 26 & 29 & {\color{red}39} & {\color{red}49} \}\\
\bar{b}^4_4= \{  4 & 14 & 17 & 27 & 30 & {\color{red}40} & {\color{red}43} \}\\
\bar{b}^4_3= \{  5 & 8 & 18 & 28 & 31 & {\color{red}41} & {\color{red}44} \}\\
\bar{b}^4_2= \{  6 & 9 & 19 & 22 & 32 & {\color{red}42} & {\color{red}45} \}\\
\bar{b}^4_1= \{  7 & 10 & 20 & 23 & {\color{red}33} & {\color{red}36} & {\color{red}46} \}\\
 \end{bmatrix},
\end{equation*}

\begin{equation*}
\bar{P}_5 = \begin{bmatrix}
\bar{b}^5_7= \{      1 & 12 & 16 & 27 & 31 & {\color{red}42} & {\color{red}46} \}\\
\bar{b}^5_6= \{ 2 & 13 & 17 & 28 & 32 & {\color{red}36} & {\color{red}47} \}\\
\bar{b}^5_5= \{  3 & 14 & 18 & 22 & {\color{red}33} & {\color{red}37} & {\color{red}48} \}\\
\bar{b}^5_4= \{  4 & 8 & 19 & 23 & {\color{red}34} & {\color{red}38} & {\color{red}49} \}\\
\bar{b}^5_3= \{  5 & 9 & 20 & 24 & {\color{red}35} & {\color{red}39} & {\color{red}43} \}\\
\bar{b}^5_2= \{  6 & 10 & 21 & 25 & 29 & {\color{red}40} & {\color{red}44} \}\\
\bar{b}^5_1= \{  7 & 11 & 15 & 26 & 30 & {\color{red}41} &{\color{red}45} \}\\
 \end{bmatrix},
\bar{P}_6 = \begin{bmatrix}
\bar{b}^6_7= \{      1 & 13 & 18 & 23 & {\color{red}35} & {\color{red}40} & {\color{red}45} \}\\
\bar{b}^6_6= \{ 2 & 14 & 19 & 24 & 29 & {\color{red}41} &{\color{red}46} \}\\
\bar{b}^6_5= \{ 3 & 8 & 20 & 25 & 30 & {\color{red}42} & {\color{red}47} \}\\
\bar{b}^6_4= \{ 4 & 9 & 21 & 26 & 31 & {\color{red}36} & {\color{red}48} \}\\
\bar{b}^6_3= \{ 5 & 10 & 15 & 27 & 32 &{\color{red}37} & {\color{red}49} \}\\
\bar{b}^6_2= \{ 6 & 11 & 16 & 28 & {\color{red}33} & {\color{red}38} & {\color{red}43} \}\\
\bar{b}^6_1= \{ 7 & 12 & 17 & 22 & {\color{red}34} & {\color{red}39} &{\color{red}44} \}\\
 \end{bmatrix},
\end{equation*}

\begin{equation*}
\bar{P}_7 = \begin{bmatrix}
\bar{b}^7_7= \{1 & 14 & 20 & 26 & 32 & {\color{red}38} & {\color{red}44} \}\\
\bar{b}^7_6= \{ 2 & 8 & 21 & 27 & {\color{red}33} & {\color{red}39} & {\color{red}45} \}\\
\bar{b}^7_5= \{ 3 & 9 & 15 & 28 & {\color{red}34} & {\color{red}40} & {\color{red}46} \}\\
\bar{b}^7_4= \{ 4 & 10 & 16 & 22 & {\color{red}35} & {\color{red}41} & {\color{red}47} \}\\
\bar{b}^7_3= \{ 5 & 11 & 17 & 23 & 29 & {\color{red}42} & {\color{red}48} \}\\
\bar{b}^7_2= \{ 6 & 12 & 18 & 24 & 30 & {\color{red}36} & {\color{red}49} \} \\
\bar{b}^7_1= \{ 7 & 13 & 19 & 25 & 31 &{\color{red}37} & {\color{red}43} \}\\
 \end{bmatrix},
 \bar{P}_8 = \begin{bmatrix}
\bar{b}^8_7= \{1 & 8 & 15 & 22 & 29 & {\color{red}36} & {\color{red}43} \}\\
\bar{b}^8_6= \{ 2 & 9 & 16 & 23 & 30 & {\color{red}37} & {\color{red}44} \}\\
\bar{b}^8_5= \{ 3 & 10 & 17 & 24 & 31 &{\color{red}38} & {\color{red}45} \}\\
\bar{b}^8_4= \{ 4 & 11 & 18 & 25 & 32 & {\color{red}39} & {\color{red}46} \}\\
\bar{b}^8_3= \{ 5 & 12 & 19 & 26 &{\color{red} 33} & {\color{red}40} & {\color{red}47} \}\\
\bar{b}^8_2= \{ 6 & 13 & 20 & 27 & {\color{red}34} & {\color{red}41} & {\color{red}48} \}\\
\bar{b}^8_1= \{ 7 & 14 & 21 & 28 & {\color{red}35} & {\color{red}42} & {\color{red}49} \}\\
 \end{bmatrix},
\end{equation*}
}
In order to construct RBD$(\widetilde{X},\widetilde{A})$,  where $|\widetilde{X}| = (q-e)(q+f)= (7-3)(7+1)= 32$ such that $\mu = 1$, we consider the following steps. 

\begin{enumerate}
\item Choose any $h = e-f = 3-1 = 2$  blocks from $\bar{P}_1$. Let these blocks be $\bar{b}^1_1$ and $\bar{b}^1_2$. Let $S_1 = \bar{b}^1_1 \cup \bar{b}^1_2 =  \{ 36, 37, 38, 39, 40, 41, 42, 43, 44, 45, 46, 47, 48, 49 \}$, as annotated in red in the above matrices.

\item  Choose another $f = 1$ block from $\bar{P}_1$. Let it be $\bar{b}^1_3$. Now choose any $e=3$ elements from it. Let these be $S_2 = \{33, 34, 35\}$. Set $S = S_1 \cup S_2 = \{33, 34, 35, 36, 37, 38, 39, 40$, $41, 42$, $43, 44, 45, 46, 47, 48, 49 \}$ (indicated in red).

\item  Remove the elements of set $S$ from  RBD$(\bar{X},\bar{A})$. Call the resulting combinatorial design RBD$(\widetilde{X},\widetilde{A})$, where $|\widetilde{X}|=32$ and $\widetilde{A} = \{\widetilde{P}_1,\widetilde{P}_2,\widetilde{P}_3,\widetilde{P}_4,\widetilde{P}_5,\widetilde{P}_6,\widetilde{P}_7,\widetilde{P}_8 \}$, presented as below.

{\tiny
$\widetilde{P}_1 = \begin{bmatrix}
\widetilde{b}^1_7= \{ 1 & 2 & 3 & 4 & 5 & 6 & 7 \} \\
\widetilde{b}^1_6= \{ 8 & 9 & 10 & 11 & 12 & 13 & 14 \}\\
\widetilde{b}^1_5= \{ 15 & 16 & 17 & 18 & 19 & 20 & 21 \}\\
\widetilde{b}^1_4= \{ 22 & 23 & 24 & 25 & 26 & 27 & 28 \}\\
\widetilde{b}^1_3= \{ 29 & 30 & 31 & 32   \}\\
\end{bmatrix},
\widetilde{P}_2 = \begin{bmatrix}
\widetilde{b}^2_7= \{ 1 & 9 & 17 & 25   \}\\
 \widetilde{b}^2_6= \{ 2 & 10 & 18 & 26 \} \\
\widetilde{b}^2_5= \{  3 & 11 & 19 & 27   \}\\
 \widetilde{b}^2_4= \{ 4 & 12 & 20 & 28 & 29  \}\\
 \widetilde{b}^2_3= \{ 5 & 13 & 21 & 22 & 30  \}\\
 \widetilde{b}^2_2= \{ 6 & 14 & 15 & 23 & 31 \}\\
 \widetilde{b}^2_1= \{ 7 & 8 & 16 & 24 & 32  \}\\
 \end{bmatrix},
 \widetilde{P}_3 = \begin{bmatrix}
\widetilde{b}^3_7= \{  1 & 10 & 19 & 28 & 30  \}\\
\widetilde{b}^3_6= \{  2 & 11 & 20 & 22 & 31 \}\\
\widetilde{b}^3_5= \{  3 & 12 & 21 & 23 & 32  \}\\
\widetilde{b}^3_4= \{  4 & 13 & 15 & 24  \} \\
\widetilde{b}^3_3= \{  5 & 14 & 16 & 25  \}\\
 \widetilde{b}^3_2= \{ 6 & 8 & 17 & 26    \}\\
\widetilde{b}^3_1= \{ 7 & 9 & 18 & 27 & 29 \}\\
 \end{bmatrix}$,
 
$\widetilde{P}_4 = \begin{bmatrix}
\widetilde{b}^4_7= \{  1 & 11 & 21 & 24  \}\\
\widetilde{b}^4_6= \{  2 & 12 & 15 & 25  \}\\
\widetilde{b}^4_5= \{  3 & 13 & 16 & 26 & 29  \}\\
\widetilde{b}^4_4= \{  4 & 14 & 17 & 27 & 30  \}\\
\widetilde{b}^4_3= \{  5 & 8 & 18 & 28 & 31 \}\\
\widetilde{b}^4_2= \{  6 & 9 & 19 & 22 & 32  \}\\
\widetilde{b}^4_1= \{  7 & 10 & 20 & 23  \}\\
\end{bmatrix},
\widetilde{P}_5 = \begin{bmatrix}
\widetilde{b}^5_7= \{      1 & 12 & 16 & 27 & 31  \}\\
\widetilde{b}^5_6= \{ 2 & 13 & 17 & 28 & 32  \}\\
\widetilde{b}^5_5= \{  3 & 14 & 18 & 22   \}\\
\widetilde{b}^5_4= \{  4 & 8 & 19 & 23   \}\\
\widetilde{b}^5_3= \{  5 & 9 & 20 & 24   \}\\
\widetilde{b}^5_2= \{  6 & 10 & 21 & 25 & 29  \}\\
\widetilde{b}^5_1= \{  7 & 11 & 15 & 26 & 30  \}\\
\end{bmatrix},
 \widetilde{P}_6 = \begin{bmatrix}
\widetilde{b}^6_7= \{      1 & 13 & 18 & 23   \}\\
\widetilde{b}^6_6= \{ 2 & 14 & 19 & 24 & 29  \}\\
\widetilde{b}^6_5= \{ 3 & 8 & 20 & 25 & 30  \}\\
\widetilde{b}^6_4= \{ 4 & 9 & 21 & 26 & 31  \}\\
\widetilde{b}^6_3= \{ 5 & 10 & 15 & 27 & 32  \}\\
\widetilde{b}^6_2= \{ 6 & 11 & 16 & 28  \}\\
\widetilde{b}^6_1= \{ 7 & 12 & 17 & 22   \}\\
\end{bmatrix}$,

$\widetilde{P}_7 = \begin{bmatrix}
\widetilde{b}^7_7= \{1 & 14 & 20 & 26 & 32  \}\\
\widetilde{b}^7_6= \{ 2 & 8 & 21 & 27  \}\\
\widetilde{b}^7_5= \{ 3 & 9 & 15 & 28   \}\\
\widetilde{b}^7_4= \{ 4 & 10 & 16 & 22   \}\\
\widetilde{b}^7_3= \{ 5 & 11 & 17 & 23 & 29  \}\\
\widetilde{b}^7_2= \{ 6 & 12 & 18 & 24 & 30  \} \\
\widetilde{b}^7_1= \{ 7 & 13 & 19 & 25 & 31  \}\\
 \end{bmatrix},
 \widetilde{P}_8 = \begin{bmatrix}
\widetilde{b}^8_7= \{1 & 8 & 15 & 22 & 29  \}\\
\widetilde{b}^8_6= \{ 2 & 9 & 16 & 23 & 30  \}\\
\widetilde{b}^8_5= \{ 3 & 10 & 17 & 24 & 31  \}\\
\widetilde{b}^8_4= \{ 4 & 11 & 18 & 25 & 32  \}\\
\widetilde{b}^8_3= \{ 5 & 12 & 19 & 26  \}\\
\widetilde{b}^8_2= \{ 6 & 13 & 20 & 27   \}\\
\bar{b}^8_1= \{ 7 & 14 & 21 & 28   \}\\
 \end{bmatrix}
$.
}
\end{enumerate}
Note that here all the blocks are not of the same sizes, but any two blocks from different parallel classes have at most 1 element in common, i.e., $\mu =1$. The blocks sizes are in the set $\{(q-3),(q-2), q \}= \{4, 5, 7\}$. The number of parallel classes in $\widetilde{A}$ remains $ q+1= 8$. This technique is now more formally explained for the general case in Construction \ref{con:1(q-e)(q+f)} below. Let $d= k \cdot s =(q-e)(q+f)$, with $0< f \leq e \leq q$. The steps for constructing the RBD$(X,A)$ are as follows.
  
\begin{cons}
\label{con:1(q-e)(q+f)}
\normalfont
Given $q$, a prime power, construct $(q^2, q, 1)$-ARBIBD. Call this design $(\bar{X}, \bar{A})$ with $\bar{X} = \{1, 2, \ldots, q^2\}$ and $|\bar{A}| = q(q+1)$ many blocks, each block is of constant size $q$. It will have $r = q+1$ many parallel classes, call them $\{\bar{P}_1,\bar{P}_2,\ldots,\bar{P}_{q+1}\}$, each parallel class having $q$ many blocks of constant size $q$. Between any two blocks from different parallel classes, exactly one element will be common, i.e., $|\bar{b}^l_i \cap \bar{b}^m_j| = 1, \forall \ l \neq m$.
\begin{enumerate}
\item  Given $e\geq f$, choose $h=e-f \geq 0$ many  blocks from $\bar{P}_1$, which are $\{\bar{b}^1_1, \bar{b}^1_2, \ldots, \bar{b}^1_h\}$. Let $S_1 =  \bar{b}^1_1 \cup \bar{b}^1_2 \cup \ldots \cup \bar{b}^1_h $. Therefore, $|S_1| = h\cdot q$. 

\item From $\{\bar{b}^1_{h+1}, \bar{b}^1_{h+2}, \ldots, \bar{b}^1_{h+f}\}$ blocks of $\bar{P}_1$, choose any $e$ number of  elements from each of them. Let $S_2$ be the union of all these elements. Therefore, $|S_2| = e\cdot f$. Let $S = S_1 \cup S_2$.  

\item Remove the elements of set $S$ from  RBD $(\bar{X},\bar{A})$ and  call the resulting design as RBD$(\widetilde{X},\widetilde{A})$.  
\end{enumerate} 
 \end{cons}
 
We claim that the above RBD$(\widetilde{X},\widetilde{A})$ is such that $|\widetilde{X}| = (q-e)(q+f)$ and $\widetilde{A}$ consists of $q+1$ many parallel classes having different block sizes, such that blocks from different parallel classes have at most one element in common, i.e., $ |b^l_i \cap b^m_j| \leq 1, \forall \ l \neq m$. Hence $\mu=1$. We formalize this in the form of a lemma below. 
 
\begin{lemma}
\label{lem:con1(q-e)(q+f)}
Let $d=(q-e)(q+f)$ for $f, e \in \mathbb{N}$ with $0 < f \leq e \leq q$ and some prime power $q$. Then one can construct an RBD$(\widetilde{X},\widetilde{A})$, with $|X| = d$ having block sizes from the set of integers $ \{(q-e), (q-e+1), \ldots, (q-e+f), q\}$ with $\mu = 1$, and having $r = q+1$ many parallel classes.
\end{lemma}
\begin{proof}
Refer to Construction \ref{con:1(q-e)(q+f)} above. Here RBD$(\bar{X}, \bar{A})$ is an ARBIBD with $|\bar{X}|= q^2$, having constant block size $q$. Note that any pair of blocks from different parallel classes have exactly one element in common, i.e., $ |\bar{b}^l_i \cap \bar{b}^m_j| = 1 \ \forall \ l \neq m$. The number of elements in the set $|S| = |S_1\cup S_2| = |S_1| + |S_2| = (e-f)q+ ef < q^2$, which is a proper subset of $\bar{X}$. These are removed from all the parallel classes of RBD$(\bar{X},\bar{A})$. Hence, the resulting design RBD$(\widetilde{X},\widetilde{A})$ is such that $|\widetilde{X}| = q^2-(e-f)q+ef = (q-e)(q+f)= d$,  having same number of parallel classes as in $\bar{A}$. The number of element common between any two blocks from different parallel classes would be at most 1, i.e., $ |\widetilde{b}^l_i \cap \widetilde{b}^m_j| \leq 1, \forall \ l \neq m$.

To obtain the sizes of the blocks in RBD$(\widetilde{X},\widetilde{A})$, note that $S_1$ contains all elements from $h = (e-f)$ number of blocks of $\bar{P_1}$. Hence $S_1$ would have at least $h$ elements in common with all the blocks of remaining parallel classes. Thus, removal of the elements in $S_1$ from the parallel classes $\bar{P}_l, l \geq 2$ will remove at least $h$ elements from each block of $\bar{P_l}$ which implies $|\bar{b}^l_i \setminus S_1| = q-(e-f)$. Further $S_2$ contains $e$ elements from $f$ many blocks of $\bar{P}_1$. Thus, the blocks in $\bar{P}_2,  \bar{P}_3, \ldots, \bar{P}_{q+1}$ will have at most $f$ many elements in common with $S_2$. Consequently, after removal of all the elements in $S$, from the parallel class $\bar{P}_l$, the block size $|\widetilde{b}^l_i|, \ l\geq 2$ will be maximum $q-(e-f) = (q-h)$ and minimum  $q-(e-f)-f = (q-e)$. Further, the blocks in $\widetilde{P}_1$ will be of sizes $q, (q-e)$ and have total $(q-h)$ number of blocks. 
\end{proof}

We will now show how one can use $f$ many blocks of the parallel class $\widetilde{P}_1$ in RBD$(\widetilde{X},\widetilde{A})$, as constructed above, to reshape any one of the parallel classes $\widetilde{P}_l, \ l\neq 1$ into a parallel class having $(q+f)$ many blocks each of size $(q-e)$, which we denote by $P_l$ and the resulting combinatorial design by RBD$(X,A)$. Since $\widetilde{P}_1$ have $q-h$ number of blocks, thus $\lfloor \frac{q-h}{f}\rfloor$ many parallel classes of $\bar{A}$ can be reshaped into a parallel class of $A$.  

Let us first demonstrate the construction by using RBD$(\widetilde{X},\widetilde{A})$, as constructed in the example earlier in this section, where $|\widetilde{X}|= (7-3)(7+1) = 4 \cdot 8= 32$, with $q= 7, e= 3$, $f= 1$ and $h=e-f=2$. Note that $\widetilde{P}_1$ has $q-h=5$ blocks, and all the other parallel classes have $q=7$ blocks each. Let us denote the excess number of elements on each block of $\widetilde{b}^l_j, \ l\geq2$ than $(q-e)$ by $m^l_j$. Hence $m^l_i = | \widetilde{b}^l_i | - (q-e)$. Here for each block of the parallel class $\widetilde{P}_l, \ l \geq 2$, the value of $m^l_i $ is either 0 or 1. If the block $|\widetilde{b}^l_i| = 5$, then $m^l_i = |\widetilde{b}^l_i| - (q-e) = 5 -(7-3) = 1$ and similarly for block $|\widetilde{b}^l_i| = 4$, $m^l_i = 0$. Note that $\sum^q_{i=1} m^l_i  =  (q-e) \cdot f = (7-3) \cdot 1 = 4, \forall \ l\geq 2$. We modify this RBD$(\widetilde{X},\widetilde{A})$ as follows.

\begin{enumerate}
\item Since $f= 1$, consider one block $\widetilde{b}^1_3$ of $\widetilde{P}_1$, which has 4 elements in it. Remove these elements from different blocks of $\widetilde{P}_2$ and add them as separate block in $\widetilde{P}_2$. Denote the resulting parallel class as $P_2$.

\item  Consider the parallel class $\widetilde{P}_3$  and the next $f=1$ block of $\widetilde{P}_1$, i.e., the block $\widetilde{b}^1_4$. Choose a block from $\widetilde{P}_3$, say $\widetilde{b}^3_1$. Since $|\widetilde{b}^3_1| = 5$, i.e., it has one element more than $q-e = 4$, hence mark one common element between $\widetilde{b}^1_4$ and $\widetilde{b}^3_1$ which is 27 in this case. Sequentially execute this for all the blocks of $\widetilde{P}_1$. This will mark the elements $\{27, 23, 22, 28\}$ on $\widetilde{b}^1_4$. 

\item  Since $m^3_j= 0$ or $1$, the above step will mark exactly $(q-e) = 4$ elements on $\widetilde{b}^1_4$. In a situation, if $m^l_i$ has more than one elements, then further iterations are required to mark exactly $(q-e)$ elements on the blocks of $\widetilde{P}_1$. Refer to Step 3 of Construction \ref{cons:2(q-e)(q+f)} later for the exact strategy in this regard.

\item  Now remove the elements marked on $\widetilde{b}^1_4$, i.e., $\{27, 23, 22, 28\}$ from $\widetilde{P}_3$ and add them as a separate block of $\widetilde{P}_3$ and denote the resulting parallel class as $P_3$.

\item  Consider the next parallel class $\widetilde{P}_4$ and the next $f=1$ block of $\widetilde{P}_1$, i.e., the block $\widetilde{b}^1_5$. Then repeat the Steps 2, 3, 4 to obtain $P_4$. 

\item Since the number of blocks in $\widetilde{P}_1  = 5$, in this way $r = \frac{5}{1} = 5$  parallel classes, i.e., $\widetilde{P}_l, \ l = 2, 3, 4, 5, 6$ can be modified. Discard $\widetilde{P}_7$ and $\widetilde{P}_8$. The resulting RBD$(X,A)$ is such that $|X|= 32$ with $A$ consisting  of parallel classes $\{P_2, P_3, P_4, P_5, P_6\}$ as shown below.
\end{enumerate}
\small
\begin{equation*}
P_2 = \begin{bmatrix}
b^2_8= \{ 29 & 30 & 31 & 32   \}\\
b^2_7= \{ 1 & 9 & 17 & 25   \}\\
 b^2_6= \{ 2 & 10 & 18 & 26 \} \\
b^2_5= \{  3 & 11 & 19 & 27   \}\\
 b^2_4= \{ 4 & 12 & 20 & 28   \}\\
 b^2_3= \{ 5 & 13 & 21 & 22   \}\\
 b^2_2= \{ 6 & 14 & 15 & 23  \}\\
 b^2_1= \{ 7 & 8 & 16 & 24  \}\\
 \end{bmatrix},
P_3 = \begin{bmatrix}
b^3_8= \{  22 & 23 & 27& 28   \}\\
b^3_7= \{  1 & 10 & 19  & 30  \}\\
b^3_6= \{  2 & 11 & 20  & 31 \}\\
b^3_5= \{  3 & 12 & 21  & 32  \}\\
b^3_4= \{  4 & 13 & 15 & 24  \} \\
b^3_3= \{  5 & 14 & 16 & 25  \}\\
b^3_2= \{ 6 & 8 & 17 & 26    \}\\
b^3_1= \{ 7 & 9 & 18  & 29 \}\\
 \end{bmatrix},
P_4 = \begin{bmatrix}
b^4_8= \{  16 & 17 & 18 & 19  \}\\
b^4_7= \{  1 & 11 & 21 & 24  \}\\
b^4_6= \{  2 & 12 & 15 & 25  \}\\
b^4_5= \{  3 & 13  & 26 & 29  \}\\
b^4_4= \{  4 & 14  & 27 & 30  \}\\
b^4_3= \{  5 & 8  & 28 & 31 \}\\
b^4_2= \{  6 & 9  & 22 & 32  \}\\
b^4_1= \{  7 & 10 & 20 & 23  \}\\
\end{bmatrix},
\end{equation*}

\begin{equation*}
P_5 = \begin{bmatrix}
b^5_8= \{ 10 & 11 & 12 & 13  \}\\
b^5_7= \{ 1  & 16 & 27 & 31  \}\\
b^5_6= \{ 2  & 17 & 28 & 32  \}\\
b^5_5= \{  3 & 14 & 18 & 22   \}\\
b^5_4= \{  4 & 8 & 19 & 23   \}\\
b^5_3= \{  5 & 9 & 20 & 24   \}\\
b^5_2= \{  6  & 21 & 25 & 29  \}\\
b^5_1= \{  7  & 15 & 26 & 30  \}\\
\end{bmatrix},
 P_6 = \begin{bmatrix}
b^6_8= \{ 2 & 3 & 4 & 5   \}\\
b^6_7= \{ 1 & 13 & 18 & 23   \}\\
b^6_6= \{  14 & 19 & 24 & 29  \}\\
b^6_5= \{  8 & 20 & 25 & 30  \}\\
b^6_4= \{ 9 & 21 & 26 & 31  \}\\
b^6_3= \{ 10 & 15 & 27 & 32  \}\\
b^6_2= \{ 6 & 11 & 16 & 28  \}\\
b^6_1= \{ 7 & 12 & 17 & 22   \}\\
 \end{bmatrix},
 P_7 = \begin{bmatrix}
 b^7_8= \{ 5 & 12 & 19 & 26  \}\\
b^7_7= \{1 & 14 & 20  & 32  \}\\
b^7_6= \{ 2 & 8 & 21 & 27  \}\\
b^7_5= \{ 3 & 9 & 15 & 28   \}\\
b^7_4= \{ 4 & 10 & 16 & 22   \}\\
b^7_3= \{  11 & 17 & 23 & 29  \}\\
b^7_2= \{ 6  & 18 & 24 & 30  \} \\
b^7_1= \{ 7 & 13  & 25 & 31  \}\\
 \end{bmatrix}.
\end{equation*}
\normalsize
We will now discuss about $P_7$ as presented above.
From the two discarded parallel classes $\widetilde{P}_7$ and $\widetilde{P}_8$, we find that one can use the block $\widetilde{b}^8_3$ and remove the elements in it from $\widetilde{P}_7$, and place them as a separate block of parallel class $\widetilde{P}_7$, resulting into another parallel class having $(q+f)=8$ many blocks each block having $(q-e)=4$ elements each. We denote this parallel class as $P_7$. Certainly this is not unique as there are other possibilities to obtain a parallel class using  $\widetilde{P}_7$ and $\widetilde{P}_8$. Thus here we actually obtain $r = 6 > \left\lfloor  \frac{q -e}{f} \right\rfloor +1 = 5$, indicating that $\left\lfloor  \frac{q -e}{f} \right\rfloor +1$ is not a tight lower bound in this example. 

The above construction is formally explained for the general case in Construction \ref{cons:2(q-e)(q+f)} below. We consider RBD$(\widetilde{X},\widetilde{A})$ constructed in \ref{con:1(q-e)(q+f)} with $|\widetilde{X}| = (q-e)(q+f)$ as the input for following construction.  To begin with, compute $m^l_j =|\widetilde{b}^l_j| - (q-e), \ l\geq2 $ for each block of $\widetilde{P}_l, \ l\geq 2$, which is the count of the excess number of elements on each block of $\widetilde{b}^l_j, \ l\geq2$ than what is required, which is $q-e$. Note that $\sum_{j=1}^q m^l_j =\sum_{j=1}^q \left(|\widetilde{b}^l_j|- (q-e)\right) = \sum_{j=1}^q |\widetilde{b}^l_j|- q (q-e) =  (q-e)(q+f)-q(q-e) = (q-e)f, \ \forall \ l \geq 2$. Thus $\sum_{j=1}^q m^l_j=(q-e)f$ is constant for all the parallel classes of RBD$(\widetilde{X}, \widetilde{A})$ except $\widetilde{P}_1$, which consists of $q-h$ blocks of sizes $\{(q-e), q\}$ and are being used to modify the $r$ number of other parallel classes and will be discarded in the end.

\begin{cons}
\label{cons:2(q-e)(q+f)}
\normalfont
Let $d= k \cdot s =(q-e)(q+f)$, with $0< f \leq e \leq q$ and we consider RBD$(\widetilde{X},\widetilde{A})$ from Construction \ref{con:1(q-e)(q+f)} with $|\widetilde{X}| = (q-e)(q+f)$  as the input. 

\begin{enumerate}
\item Consider  $f$ many blocks of $\widetilde{P}_1$, which has $(q-e)$ elements, i.e., the blocks $\{\widetilde{b}^1_{h+1}, \widetilde{b}^1_{h+2}, \ldots$, $\widetilde{b}^1_{h+f}\}$. Remove the elements in the blocks $\{\widetilde{b}^1_{h+1}, \widetilde{b}^1_{h+2}, \ldots,\widetilde{b}^1_{h+f}\}$ from different blocks of $\widetilde{P}_{2}$ and add the blocks $\{ \widetilde{b}^1_{h+1}, \widetilde{b}^1_{h+2},\ldots, \widetilde{b}^1_{h+f}\}$ as blocks of $\widetilde{P}_{2}$. Name the resulting parallel class as $P_2$. 

\item Consider the parallel class $\widetilde{P}_{3}$ and next $f$ many blocks of $\widetilde{P}_1$, i.e., $ \{ \widetilde{b}^1_{(h+f+1)}, \widetilde{b}^1_{(h+f+2)}, \ldots$, $\widetilde{b}^1_{(h+2f)}\}$, each consisting of $q$ many elements. Call this set of blocks as $S^1_3$. Select a block from $\widetilde{P}_{3}$, say $\widetilde{b}^3_1$. Corresponding to this block, mark $m^3_1$ number of elements which are common with the blocks in set $S^1_3$. Then move them to the next block of $\widetilde{P}_{3}$, namely $\widetilde{b}^3_2$, and mark $m^3_2$ number of elements common with the blocks in $S^1_3$. Sequentially continue this for all the blocks of $\widetilde{P}_{3}$. 

\item Now consider $\widetilde{b}^1_u, \ \widetilde{b}^1_v \in S^1_3$, such that $\widetilde{b}^1_u$ has more elements marked than $(q-e)$ and $\widetilde{b}^1_v$ has less elements marked than $(q-e)$. Identify the blocks of $\widetilde{P}_3$ which have a marked element common with $\widetilde{b}^1_u$,  but has an unmarked element common with $\widetilde{b}^1_v$, say block $\widetilde{b}^3_j$. Then unmark this element in $\widetilde{b}^1_u$, and mark the common element between the block $\widetilde{b}^3_j$ and $\widetilde{b}^1_v$ on the block $\widetilde{b}^1_v$. Thus,  the marked element on $\widetilde{b}^1_u$ is removed and the marked on $\widetilde{b}^1_v$ is added. Iterate this for all such blocks in $S^1_3$ which has marked elements different from $(q-e)$ and continue this till all the blocks in set  $S^1_3$ have exactly $(q-e)$ marked element. 

Later in Lemma \ref{lem:con(q-e)(q+f)}, we will show that such a block $\widetilde{b}^3_j$ will always exist. Further this process will terminate in a finite number of steps as there are finite number of blocks and elements, and every iteration adds the marked element of a block having less elements marked than $(q-e)$ and removes the mark element of a block having  more elements marked than $(q-e)$.

\item Remove the elements marked on each of the blocks in the set $S^1_3$ from the parallel class $\widetilde{P}_{3}$ and add the elements marked on the block, say  $\widetilde{b}^1_{h+f+1}$ as separate blocks in $\widetilde{P}_{3}$. Similarly  add elements marked on the block $\widetilde{b}^1_{h+f+2}$ as separate blocks in $\widetilde{P}_{3}$ and so on for all the blocks in $S^1_3$. Call the resulting parallel class as $P_3$. 

\item  Then consider the next parallel class $\widetilde{P}_{4}$ and the next $f$ many blocks of $\widetilde{P}_1$ and repeat the steps 2, 3 and 4 as mentioned above. Denote the resulting class as $P_4$ and continue till the number of blocks in $\widetilde{P}_1$ becomes less than $f$ . 

\item Since the number of blocks in $\widetilde{P}_{1}$ is $q-h = q-(e-f)$, in this way $r = \left\lfloor  \frac{q - (e-f)}{f} \right\rfloor = \left\lfloor  \frac{q - e}{f} \right\rfloor +1$ many parallel classes $\widetilde{P}_{r}$ can be modified. Then the RBD$(X, A)$, where $X = \{1, 2, \ldots, (q-e)(q+f) \}$ and $A =\{P_2, P_3,\ldots, P_{r+1}\}$ is the required design.
\end{enumerate}
\end{cons}
Note that in Step 4 above, since each block in the set $S^1_3$ has $q-e$ elements marked and total number of blocks is $f$, hence total number of elements removed from $\widetilde{P}_3$ are $(q-e)\cdot f$. Thus, we are basically using a set of $f$ blocks from the parallel classes $\widetilde{P}_1$ to reshape one of the remaining parallel class, into blocks of size $(q-e)$, having $(q+f)$ blocks.  This is achieved by identifying $(q-e)$ elements of one block from the parallel classes $\widetilde{P}_1$, and in different blocks of  parallel class say $\widetilde{P}_l$. Thereafter, we delete these elements from different blocks $\widetilde{P}_l$, and add these elements as a separate block in $\widetilde{P}_l$. Thus, the resulting parallel class will consist of $(q+f)$ blocks, each having $(q-e)$ elements.
 
We claim that the above design $(X,A)$ is a RBD, such that $|X| = (q-e)(q+f)$ and $A$ consists of $r= \left\lfloor  \frac{q - e}{f} \right\rfloor +1$ many parallel classes each having $(q+f)$ many blocks  each of size  $(q-e)$, such that the blocks from different parallel classes have at most one element in common, i.e., $|b^l _i \cap b^m_j | \leq 1 \ \forall \ i \neq j$. We formalize this in the following lemma.

\begin{lemma}
\label{lem:con(q-e)(q+f)}
Let $d=(q-e)(q+f)$, for $f, e \in \mathbb{N}$ with $0 < f \leq e \leq q$ and $q$ some power of prime. Then one can construct an RBD$(X, A)$, with $|X| = d$ having constant block size $(q-e)$ with $\mu = 1$, and having at least $r =  \left\lfloor  \frac{q -e)}{f} \right\rfloor+ 1$ many parallel classes.
\end{lemma}
\begin{proof}
Refer to Construction \ref{cons:2(q-e)(q+f)} above. Since in RBD$(\widetilde{X},\widetilde{A})$  any pair of blocks from different parallel classes has at most one element in common, we have $|\bar{b}^l_i \cap \bar{b}^m_j| \leq 1 \ \forall \ l \neq m$. Since no element of RBD$(\widetilde{X},\widetilde{A})$ has been deleted or added to it to obtain RBD$(X, A)$, hence $d=|\widetilde{X}|=|X|= (q-e)(q+f)$. From the Step 6 of Construction \ref{cons:2(q-e)(q+f)}, we obtain $r =  \left\lfloor  \frac{q -e)}{f} \right\rfloor+ 1$.

Now we show that Step 2 can be successfully executed. That is, from the set of $f$ many blocks of $\widetilde{P}_1$, it will be possible to mark $\sum^q_{i=1} m^l_i  =  (q-e)\cdot f, \ \forall \ l\geq 2$. Note that each block  $\widetilde{b}^1_j, \ j> e$ has exactly one element common with $\widetilde{b}^l_j, \ j= 1, 2, \ldots, q, \ l\geq 2$. As $0 \leq m^l_j \leq f $, corresponding to each block $\widetilde{b}^l_j$, there will always be $m^l_j$ elements on different blocks, which is total $f$ in number. Hence Step 2 can be successfully executed.
    
Steps 3 and 4 are related to the construction where elements are to be marked on blocks in the set $S^1_3$ having $f$ many blocks. These steps will finally result into $(q-e)$ elements being marked on each of these blocks. For this, note that $\sum_{i=1}^q m^l_i =  (q-e)f $. That means, if it is not possible to mark $q-e$ elements on each of the blocks in the set $S^1_3$, then on some block there will have more elements marked than $q-e$ and on some blocks there are less element marked than $q-e$. It is not possible that all the blocks have less than $(q-e)$ elements marked or all the blocks have more than $(q-e)$ elements marked as in that case $\sum_{i=1}^q m^l_i <  (q-e)f $ or $> (q-e)f$ accordingly.
 
In case there is a block $\widetilde{b}^1_u, \ \widetilde{b}^1_v \in S^1_3$, such the $\widetilde{b}^1_u$ has more elements marked than $(q-e)$ and $\widetilde{b}^1_v$ has less element marked than $(q-e)$, then there will exist a block of $\widetilde{P}_3$ which have marked element common with $\widetilde{b}^1_u$,  but non-marked element common with $\widetilde{b}^1_v$. Suppose there is no such block in $\widetilde{P}_3$. Then the marked elements of block $\widetilde{b}^1_u$ (which is $> q-e$) and the unmarked elements of $\widetilde{b}^1_v$ (which is $> e$) would all lie on the different blocks of $\widetilde{P}_3$. However, this would imply number of blocks in $|\widetilde{P}_3| > q-e+e = q$, which is a contradiction as $|\widetilde{P}_3| = q$. 

Finally, we show that $\mu=1$. Note that the blocks which has been added in the parallel classes $\widetilde{P}_2, \widetilde{P}_3, \ldots, \widetilde{P}_{r+1}$ to construct the parallel classes $P_2, P_3, \ldots, P_{r+1}$ are respectively part of the blocks of $\widetilde{P}_1$. Since any block of $\widetilde{P}_1$ has at most one element common with any other block of $\widetilde{P}_l, \ l\geq 2$, $\mu$ will be 1 for RBD$(X,A)$.
\end{proof}

Now we can use this RBD$(X,A)$ to construct APMUBs in dimension $d= |X|= (q-e)(q+f)$ having parameters as given by Theorem~\ref{th1}, which we formally state and prove in the next theorem.

\begin{theorem}\label{th:(q-e)(q+f)}
Let $d=(q-e)(q+f)$, for some prime power $q$, where $0 < f \leq e$ and $0 < (e+f) \leq \frac {3}{2} d^{\frac{1}{2}}$, $e, f \in \mathbb{N}$.  Then there exist at least $r =  \left\lfloor  \frac{q - e}{f} \right\rfloor +1$ many APMUBs, with $\Delta = \{0, \frac{1}{q-e}\}$, $\beta = \sqrt{\frac{q+f}{q-e} } =1 - \mathcal{O}(d^{-\lambda}) \leq 2$, where $\lambda= \frac{1}{2} $ and $\epsilon = 1 - \frac{1}{q+f}$. Further, if there exists a Real Hadamard matrix of order $(q-e)$, then one can construct $r$ many Almost Perfect Real MUBs with same parameters.
\end{theorem}
\begin{proof}
In order to show that we can produce such APMUBs, let us consider an RBD$(X,A)$ with $|X| = (q-e)(q+f)$ having $r$ parallel classes of constant block size $(q-e)$, such that between the blocks from different parallel classes, there is at most one element in common, and hence $\mu = 1$.Then using this RBD along with a Hadamard matrix of order $(q-e)$, we can construct orthonormal bases following Theorem~\ref{th1}. Further, the condition that $\beta < 2$ for APMUB gives $(e+f) \leq \frac{3}{2} d^{\frac{1}{2}}$. In terms of $q$, this inequality becomes  $4e+f \leq 3 q$. The parameters of APMUBs are $\Delta = \{0, \frac{1}{q-e} \}$, $\beta = \sqrt{\frac{q+f}{q-e}} = 1 + \mathcal{O}(d^{-\frac{1}{2}}) \leq 2$ and sparsity $\epsilon = 1- \frac{1}{q+f}$. Further when a real Hadamard matrix of order $(q-e)$ is available, the same can be used in the construction of Approximate Real MUBs. Since number of parallel classes $r $ is at least $\lfloor\frac{q-e}{f} \rfloor+1$, hence we get at least these many APMUBs.
\end{proof}
\begin{remark}
\label{rem5x}
Note that $\frac{q - e}{f}$ is $\mathcal{O}(\sqrt{d})$, when $e, f$ are considered to be constants. That is, in such cases we obtain 
$\mathcal{O}(\sqrt{d})$ many APMUBs for the dimension $d$.
\end{remark}
Following Theorem~\ref{th:(q-e)(q+f)}, we can get at least $r =  \left\lfloor  \frac{q - e}{f} \right\rfloor +1$ many APMUBs. This can enable us to beat the {\sf{Mutually Orthogonal Latin Square (MOLS) Lower Bound construction}} for APMUBs (Theorem \ref{th:(s-e)s}), according to which we obtain $N(q+f)+1$ many APMUBs. Let us present a few illustrative examples in this regard. 
\begin{itemize}
\item For $d = 60 = 6\cdot 10$, the known value of $N(10)$ is $2$ hence {\sf{MOLS Lower Bound construction}} provides three APMUBs with $\beta$ value of 1.29. On the other hand, if we use above construction method, by expressing $d  = (9-3)(9+1)$, we obtain $\frac{9-(3-1)}{1}= 7$ many APMUBs with $\beta= 1.29$. In this case the number of complex MUBs, that can be constructed following prime factorization formula, is $3+1 = 4$ only.
\item For $d = 24 =  4 \cdot 6$, with $N(6) = 1$, the {\sf{MOLS Lower Bound construction}} generates $2$ APRMUBs with $\beta = 1.22$. On the other hand, expressing $d= 24 = (5-1)(5+1)$ we obtain 5 APRMUBs with $\beta = 1.22$. The number of real MUBs for $d=24$ is 2 \cite{boykin2005real} only.
\end{itemize}

Further, to illustrate the advantage of Construction \ref{cons:2(q-e)(q+f)} over Construction \ref{cons:(s-e)s}, consider the example of $d= 2^2 \cdot 3^2 \cdot 5 \cdot 7$  and different ways of expressing it as product of two factors, that we considered earlier following Theorem \ref{th:(s-e)s}. 
\begin{itemize}
\item Expressing  $d = 30 \cdot 42$, the {\sf{MOLS Lower Bound construction}} will provide $N(42)+ 1 = 6$ many APMUBs with $\beta = 1.18$, where as expressing $d= (41-11)(41+1)$ and using Theorem \ref{th:(q-e)(q+f)} will provide 31 many APMUBs with $\beta = 1.18$.
\item or expressing  $d = 28 \cdot 45$, the {\sf{MOLS Lower Bound construction}} will provide $N(45)+ 1 = 7$ many  APRMUB with $\beta = 1.27$, where as expressing $d= (43-15)(43+2)$ and using Theorem \ref{th:(q-e)(q+f)} will provide 15 many  APRMUB with $\beta = 1.27$. 
\end{itemize}
Note that, $N(42) = 5$ and $N(45) = 6$ are the presently known values of the maximum number of MOLS of these orders \cite{HandBook}. Here expressing $d= 35 \cdot 36$, we cannot use Theorem \ref{th:(q-e)(q+f)} as it cannot be expressed as $(q-e)(q+f)$, with $q$ some power of prime and $0 \leq f\leq e $. Thus with this factorization of $d$, {\sf{MOLS Lower Bound construction}} provides $N(36)+ 1 = 9$ APMUBs with $\beta =1.01$. 

Note that, $r =  \left\lfloor  \frac{q - (e-f)}{f} \right\rfloor$, with condition $0 < f\leq e$, is maximum for a given $q$ when $f = e = 1$. From the asymptotic expression of $\beta$, it is clear that for $0< f\leq e $ we will obtain $\beta$ closest to 1 when $e= f=1$. Thus for $e = f = 1$, we state the result of APMUB as a corollary below.
\begin{corollary}
\label{cor:(q-1)(q+1)}
Let $d = q^2-1 = (q-1)(q+1)$ where $q$ is a prime power. Then one can construct $q$ many APMUBs with $\Delta = \left\{0, \frac{1}{q-1} \right\}$ and $\beta =  \sqrt{\frac{q+1}{q-1}}$ with sparsity $\epsilon = 1- \frac{1}{q+1}$. If a real Hadamard Matrix of order $q-1$ exists, then we have $q$ many APRMUBs with the same parameters. 
\end{corollary}
Our observation made in connection with Hadamard matrix of order $(q-1)$, constructed through Paley method \cite{paley1933orthogonal} after Corollary \ref{cor:(q-1)q} is applicable here as well. We revisit that once again here for better explanation. That is, if $m = \frac{q-3}{2} \equiv 1 \bmod 4$ and $m$ is some prime power, then using the Paley Construction \cite{paley1933orthogonal}, we can construct Hadamard matrix of order $2(m+1) = q-1$. Thus, for a prime power $q \equiv 1 \bmod 4$, if $\frac{q-3}{2}$ is also some prime power, as well as, is equivalent to $1 \bmod 4$, then the real Hadamard matrix of order $q-1$ will exist through the Paley Construction. For all such $q$'s, there exist $q$ APRMUBs in $\mathbb{R}^{q^2 - 1}$ and the above corollary will become independent of the Hadamard Conjecture. Examples of such $q$ are $13, 29, 5^3$ etc. 
\begin{itemize}
\item For $d = (13-1)(13+1)= 2^3\cdot 3 \cdot 7$, we obtain $13$ many APRMUBs with $\beta = 1.080$. In this case number of real MUBs is only 2 and complex MUBs is 4.
 \item For $d= (29-1)(29+1) = 2^3 \cdot 3\cdot 5 \cdot 7$, we obtain $29$ many APRMUBs with $\beta = 1.035$. In this case also the number of real MUBs is only 2 and complex MUBs is 4.
\end{itemize}
The above examples clearly indicates that as $d$ increases, $\beta$ approaches closer to 1, hence we obtain APRMUBs which are significantly close to the MUBs.  

\subsection{Some problems that require further attention}
\label{further5}
It was pointed out in the example constructed above for RBD$(X, A)$, that for $|X|= d= 4\cdot 8 = (7-3)(7+1)$, one could construct more number of parallel classes than $r = \left\lfloor  \frac{q - e}{f} \right\rfloor +1 = 5$ in this case, $q=7, e=3, f=1$. From our experience of constructing RBD$(X,A)$, for the situation when $|X|$ can be expressed as $(q-e)(q+e) = q^2 -e^2$, i.e., for the situation $e=f> 0$, there appears to be always more than $r =  \left\lfloor  \frac{q}{f} \right\rfloor +1$ many parallel classes. In this situation it is possible to use other parallel classes, apart from the first one of $(q^2,q,1)$-ARBIBD, which enable us to obtain more parallel classes for RBD$(X,A)$ than $r =  \left\lfloor  \frac{q}{f} \right\rfloor+1$. 
A proof of this in the following form in a general setting might be an interesting open problem.

{\it Let $d=(q-e)(q+e)$, for $e \in \mathbb{N}$ with $0  \leq e \leq q$ and $q$ a prime power. Then one can construct an RBD$(X,A)$, with $|X| = d$ having constant block size $(q-e)$ with maximum intersection number $\mu = 1$, and having $r \geq \frac{q}{2}$ many parallel classes.}

Further our efforts for the following form of composite $d$  could not result into number of APMUBs of the order of $\mathcal{O}(\sqrt{d})$, where $q$ is a prime power.
 
\begin{itemize}
\item For $d= q(q+f)$, RBD having block size $q$, with $q+f$ many blocks in each parallel classes.
\item For $d = (q-e)(q+f),  0 < e < f$, RBD having block size $q-e$ with $q+f$ many blocks in each parallel classes. 
\end{itemize}
For the above forms of $d$, we could not construct more number of APMUBs than what is given by {\sf{Mutually Orthogonal Latin Square Lower Bound construction}}. Further efforts in this direction or new ideas may be required for this. We believe that it should be possible to improve the MOLS lower bound in such cases as well. 

Note that our construction of APMUBs are very sparse and hence the set of Bi-angular vectors  are very sparse. The sparsity of each vector inner constructions, is $\epsilon = 1- \frac{1}{s} \approx 1- \frac{1}{\sqrt{d}}$. Also the non zero components of the vectors are all of the same absolute value, which is  $ \frac{1}{s-e}\approx \frac{1}{\sqrt{d}}  $. Our extensive search of literature could not find any study on bounds on the cardinality of such kind of sparse vectors, each having same sparsity. Nevertheless, there are bounds on the cardinality of flat equiangular lines. Here flat signifies that all the component of the vectors are of same magnitude. In such situation the cardinality of set of equiangular lines in $\mathbb{C}^d$ is bounded by $(d^2 - d -1)$, Refer to \cite[Lemma 2.2]{godsil2009equiangular} which is less than $d^2$, which is cardinality when the constrain of flatness is relaxed. We similarly believe that the cardinality of such Bi-angular set, with with such large sparsity would be significantly less than those given in~\cite[Table I]{delsarte1991bounds} and~\cite[Equations 3.9,5.9]{Calderbank1997}]. Hence we subsequently intend to study the bounds on the cardinality of the set of Bi-angular vectors, with large sparsity.

\section{Conclusion}
\label{conclusion}
In this paper we consider construction of APMUBs, which are significantly good approximation of MUBs. In asymptotic sense, the APMUBs are almost as good as the MUBs. That is, for a dimension $d$, the value of the dot product between two vectors from different bases will be very close to $\frac{1}{\sqrt{d}}$, and in a few cases 0. In this paper we have formalized the definition of APMUBs and shown that for a good proportion of integers, we can construct $O(\sqrt{d})$ many APMUBs. Such a generic result is elusive in cases of perfect MUBs. Thus, for all practical purposes in the domain of quantum information, or related areas, our construction ideas open up a larger possibility of obtaining required combinatorial structures. How dense are these values of $d$ for which we can construct such APMUBs? We note by preliminary calculations that for $d < N$, there are approximately $\frac{3}{4}\frac{N}{\log N}$ many such $d$'s, which can be expressed as $(q-e)q$ with $0 \leq e< \frac{3}{4}q$. One may note that this is of the order of the density of the primes. As the main scope of this paper is understanding the combinatorial techniques, we leave this as a future research effort. Another important issue in this regard is that our constructions are directly related to the concept of Bi-angular vectors. We primarily note that when two vectors are randomly selected from the set of such Bi-angular vectors, there exists a very large probability that they will be making an angle of $\frac{\beta}{\sqrt{d}}$. In fact as $d$ increases, the probability converges to certainty. This is the scenario that happens in our APMUB related constructions. We leave this too for future investigation in a disciplined manner.

\section{Declarations: Funding and/or Competing interests}
No funding was received to assist with the preparation of this manuscript. 
The authors have no relevant financial or non-financial interests to disclose.
The authors have no conflict in any form related to this manuscript.

\bibliographystyle{plain}

\end{document}